\documentclass[draft,a4paper,11pt]{article}
\usepackage[final,breaklinks]{hyperref}
\usepackage[final,breaklinks]{xurl}

\usepackage[margin=2.5cm,includefoot]{geometry}
\usepackage{adjustbox}
\usepackage[nosort]{cite}
\usepackage[british]{babel}
\usepackage{orcidlink}
\usepackage{xcolor}
\definecolor{MyDarkBlue}{rgb}{0.15,0.25,0.85}
\hypersetup {
    pdftitle = {Full S‐matrices and Witten diagrams with (relative) L infinity‐algebras},
    pdfauthor = {Luigi Alfonsi, Leron Borsten, Hyungrok Kim, Martin Wolf, and Charles Young},
    pdfsubject = {hep-th, math-ph},
    pdflang = {en-GB},
    pdfkeywords = {L infinity-algebra, minimal model, S-matrix, holography, AdS/CFT, BV–BFV formalism, Batalin–Vilkovisky formalism},
    colorlinks = true,
    citecolor = MyDarkBlue,
    linkcolor = MyDarkBlue,
    urlcolor = MyDarkBlue,
    breaklinks = true,
    linktocpage = true,
    hypertexnames = false
}
\usepackage{amsmath,amssymb,amsthm,mathtools,booktabs,pdflscape,mleftright,tikz-cd,tikz-feynhand}
\usetikzlibrary{decorations.pathreplacing,decorations.markings,snakes}
\usepackage[osf]{newpxtext}
\usepackage[vvarbb,slantedGreek]{newpxmath}
\usepackage[final]{microtype}
\usepackage[textsize=tiny]{todonotes}

\ExplSyntaxOn
\NewCommandCopy\latextextsection\textsection
\RenewDocumentCommand{\textsection}{}
 {
  \str_if_eq:eeTF { \use:c { f@family } } { \familydefault }
   {
    {\fontfamily{qpl}\selectfont\latextextsection}
   }
   {\latextextsection}
 }
\ExplSyntaxOff

\let\OLDthebibliography\thebibliography
\renewcommand\thebibliography[1]{
  \OLDthebibliography{#1}
  \setlength{\parskip}{0pt}
  \setlength{\itemsep}{0pt plus 0.3ex}
}

\tikzset{>=stealth}
\tikzcdset{arrow style=tikz}

\usepackage{cleveref}
\theoremstyle{definition}
\newtheorem{definition}{Definition}
\theoremstyle{plain}
\newtheorem{proposition}[definition]{Proposition}

\theoremstyle{definition}

\theoremstyle{remark}
\newtheorem{remark}[definition]{Remark}

\renewcommand\delta\deltait
\newcommand\LV{V}

\newcommand\email[1]{\href{mailto:#1}{\nolinkurl{#1}}}
\renewcommand\Delta\Deltaup
\renewcommand\pi\piup
\renewcommand\Gamma\Gammaup
\renewcommand\Omega\Omegaup
\newcommand{\be}{\begin{equation}}
\newcommand{\ee}{\end{equation}}

\title{Full $S$-Matrices and Witten Diagrams\\ with Relative $L_\infty$-Algebras}

\author{Luigi Alfonsi\,\orcidlink{0000-0001-5231-2354}\,\footnote{Department of Physics, Astronomy and Mathematics, University of Hertfordshire, Hatfield \textsc{al10 9ab}, United Kingdom}~,~~Leron~Borsten\,\orcidlink{0000-0001-9008-7725}\,\footnotemark[1]~,~~Hyungrok~Kim\,\orcidlink{0000-0001-7909-4510}\,\footnotemark[1]~,\\ Martin~Wolf\,\orcidlink{0009-0002-8192-3124}\,\footnote{School of Mathematics and Physics, University of Surrey, Guildford \textsc{gu2 7xh}, United Kingdom}~,~~and Charles A.~S.~Young\,\orcidlink{0000-0002-7490-1122}\,\footnotemark[1]~~\footnote{E-mail addresses: \email{l.alfonsi@herts.ac.uk}, \email{l.borsten@herts.ac.uk}, \email{h.kim2@herts.ac.uk}, \email{m.wolf@surrey.ac.uk}, \email{c.young8@herts.ac.uk}}}

\begin{document}

    \maketitle

        \begin{abstract}
            The $L_\infty$-algebra approach to scattering amplitudes elegantly describes the nontrivial part of the $S$-matrix but fails to take into account the trivial part. We argue that the trivial contribution to the $S$-matrix should be accounted for by another, complementary $L_\infty$-algebra, such that a perturbative field theory is described by a cyclic relative $L_\infty$-algebra. We further demonstrate that this construction reproduces Witten diagrams that arise in AdS/CFT including, in particular, the trivial Witten diagrams corresponding to CFT two-point functions. We also discuss Chern--Simons theory and Yang--Mills theory on manifolds with boundaries using this approach.
        \end{abstract}

    \newpage

    \tableofcontents

    \bigskip
    \bigskip
    \hrule
    \bigskip
    \bigskip

    \section{Introduction and results}\label{sec:intro}
    
    While Lagrangians determine tree-level scattering amplitudes, it is well known that sometimes different Lagrangians give equivalent scattering amplitudes. Put differently, Lagrangians carry redundant information. The on-shell--scattering--amplitudes programme sheds this redundancy by constructing the  amplitudes directly, often removing the necessity of a Lagrangian altogether. This has led to both powerful computational tools and many new insights (see~\cite{Travaglini:2022uwo} and references therein for reviews). 

    A complementary approach is to identify a natural notion of equivalence between Lagrangians that encode the same physics. The appropriate equivalence relation, adopted here, amongst Lagrangians is given in terms of quasi-isomorphisms between the cyclic $L_\infty$-algebras governing physical theories.

    The starting point of this picture is the observation that field theories described by actions correspond to cyclic $L_\infty$-algebras. In a nutshell, a cyclic $L_\infty$-algebra\footnote{For notational clarity, we will write $(\LV,\{\mu_n\}_{n\in\mathbb{N}},\langle-,-\rangle_\LV)$ as $(\LV,\mu_n,\langle-,-\rangle_\LV)$.}, $\mathfrak{L}\coloneqq(\LV,\{\mu_n\}_{n\in\mathbb{N}},\langle-,-\rangle_\LV)$, consists of a graded vector space $\LV=\bigoplus_{k\in\mathbb{Z}}\LV^k$ equipped with higher $n$-ary brackets, $\mu_n:\LV\times\cdots\times\LV\rightarrow\LV$, and cyclic inner-product (or structure)  $\langle-,-\rangle_\LV:\LV\times\LV\rightarrow\mathbb{R}$, generalising the binary Lie bracket and Cartan--Killing form of Lie algebras. The cyclic structure and higher brackets canonically yield a homotopy Maurer--Cartan action, which is precisely the classical Batalin--Vilkovisky (BV) action of the associated theory.  Roughly speaking, the cyclic structure $\langle-,-\rangle_\LV$ is an inner-product on the space of (anti)fields that yields the action\footnote{Typically integration over the spacetime manifold together with an invariant inner-product on the space of internal symmetry representations carried by the fields.}, the unary bracket $\mu_1$ encodes the kinetic term of the action, given by $\langle\phi,\mu_1(\phi)\rangle_\LV$, while the higher brackets $\mu_n$ encode the $(n+1)$-point interaction terms, given by $\langle\phi,\mu_n(\phi,\ldots,\phi)\rangle_\LV$. 

    Cast in this language, semi-classical equivalence between physical theories then amounts to quasi-isomorphisms of cyclic $L_\infty$-algebras. This language is fruitful in that the previous sentence is but the very beginning of a large and detailed dictionary between physics and homotopy algebras outlined in \Cref{table:physics-homotopy-dictionary}. In particular, every $L_\infty$-algebra $\mathfrak{L}$ is equivalent (i.e.~quasi-isomorphic) to a unique\footnote{up to $L_\infty$-isomorphisms} $L_\infty$-algebra $\mathfrak{L}^\circ$ which has vanishing $\overset\circ\mu_1$. A representative of this equivalence class  is called a \emph{minimal model}. When the cyclic structure is preserved by the quasi-isomorphism, the higher brackets of the minimal model encode precisely the non-trivial part of the corresponding connected $S$-matrix. From this perspective, off-shell Lagrangians and scattering amplitudes are unified as quasi-isomorphic $L_\infty$-algebras. This makes it clear that quasi-isomorphisms are the correct notion of equivalence; all quasi-isomorphic Lagrangians are quasi-isomorphic to the same minimal model encoding the unique $S$-matrix.

    However, there is an important gap in the above dictionary. The cyclic $L_\infty$-algebra only contains information about the nontrivial part of the $S$-matrix; information about the identity part of the $S$-matrix must be supplied separately. Lacking the identity component, one cannot simply exponentiate the connected diagrams --- after all, the minimal model, by definition, has no $\mu_1$ and, hence, no `propagator' for the identity part of the $S$-matrix. For Minkowski space-time, this is not a serious loss, of course, since it is the literal identity. However, in the case of perturbation theory on nontrivial spaces such as anti-de~Sitter space, this is no longer true and the `trivial' Witten diagrams encode nontrivial information such as the CFT two-point function.

    This observation is closely related to more technical aspects of the homotopy Maurer--Cartan action in the presence of boundaries. Firstly, the putative cyclic structure, although well-defined in the absence of a boundary, may fail to be cyclic due to boundary contributions that appear when using integration by parts to establish the required cyclic identities. For the same reason, the canonical homotopy Maurer--Cartan action may differ from the physically preferred  action, even if the cyclic structure is well-defined. For example, the canonical homotopy-Maurer--Cartan scalar-field-theory kinetic term is
    \begin{equation}
        S_\text{hMC}=\tfrac12\langle\phi,\mu_1(\phi)\rangle+\cdots=\tfrac12\int_M\operatorname{vol}_M\phi\Delta\phi +\cdots,
    \end{equation} 
    which differs from the physically relevant action functional $\frac12\int_M\operatorname{vol}_M(\partial\phi)^2$ by a boundary term coming from the total derivative; indeed, this term is precisely that which appears in the computation of the CFT two-point function. Generically, for technical reasons, in the $L_\infty$-algebraic approach to perturbation theory~\cite{Nutzi:2018vkl,Macrelli:2019afx,Arvanitakis:2019ald,Jurco:2019yfd,Lopez-Arcos:2019hvg,Jurco:2020yyu,Saemann:2020oyz,Gomez:2020vat,Bonezzi:2023xhn}, the cyclic structure is defined piecewise for the on-shell, $\mu_1(\phi)=0$, and off-shell, $\mu_1(\phi)\not=0$, components. So, if we  decompose $\phi$ it into off-shell and on-shell components $\phi\eqqcolon F+f$ respectively, the linearised homotopy Maurer--Cartan action is
    \begin{equation}
        S_\text{hMC}=\tfrac12\langle\phi,\mu_1(\phi)\rangle=\tfrac12\langle F,\mu_1(F)\rangle,
    \end{equation}
    so that there is no contribution from the on-shell part $f$. However, as we have seen in the scalar field example, the on-shell part can  contribute to the physically relevant action functional.

    All these observations are uniformly addressed by transitioning from a cyclic $L_\infty$-algebra to a \emph{relative} cyclic $L_\infty$-algebra that induces a canonical \emph{relative} homotopy Maurer--Cartan action, which form the central part of this work. A relative cyclic $L_\infty$-algebra is simply a pair of cyclic $L_\infty$-algebras, $\mathfrak{L}$ and $\mathfrak{L}_\partial$, with a cyclic morphism $\pi:\mathfrak{L}\rightarrow\mathfrak{L}_\partial$ relating them. This is the homotopy relaxation, via Koszul duality, of relative metric Lie algebras, i.e.~pairs of homomorphic Lie algebras $\pi:\mathfrak{g}\rightarrow\mathfrak{g}_\partial$ equipped with inner-products preserved by $\pi$. The key idea is to supplement the original cyclic $L_\infty$-algebra $\mathfrak{L}$, the `bulk', with another cyclic $L_\infty$-algebra $\mathfrak{L}_\partial$, the `boundary'. The raison d'\^etre of the boundary $L_\infty$-algebra is to simultaneously correct the failure of cyclicity while introducing the physically relevant boundary terms that are not present in the canonical homotopy Maurer--Cartan action. In particular, the bulk-to-boundary morphism $\pi$ generates the boundary terms in the relative homotopy Maurer--Cartan action. Every relative $L_\infty$-algebra $\pi:\mathfrak{L}\rightarrow\mathfrak{L}_\partial$ is equivalent to a minimal relative $L_\infty$-algebra $\overset\circ\pi:\mathfrak{L}^\circ\rightarrow\mathfrak{L}^\circ_\partial$. As before, the higher brackets, $\overset\circ\mu_k$, of the minimal model encode the non-trivial connected $S$-matrix, while the `trivial' part is recovered  from boundary contributions to the relative minimal model given by $\overset\circ\pi$. In conclusion, we thus arrive at an abstract structure, relative cyclic $L_\infty$-algebras, that encodes the physics associated to (asymptotic) boundaries, from $S$-matrices to Witten diagrams, uniformly.  
    
    \begin{table}[h]
        \vspace{15pt}
        \begin{center}
            \resizebox{.95\hsize}{!}{
                \begin{tabular}{cc}
                    \toprule
                    Homotopy Algebras & Scattering Amplitudes
                    \\
                    \midrule
                    homotopy Maurer--Cartan action for $L_\infty$-algebras~\cite{Zeitlin:2007yf,Zeitlin:2008cc,Hohm:2017pnh,Jurco:2018sby,Jurco:2019bvp} & perturbative field theory actions without boundary terms 
                    \\
                    minimal model of $L_\infty$-algebras~\cite{Nutzi:2018vkl,Macrelli:2019afx,Arvanitakis:2019ald,Jurco:2019yfd,Lopez-Arcos:2019hvg,Arvanitakis:2020rrk,Jurco:2020yyu,Saemann:2020oyz,Gomez:2020vat,Arvanitakis:2021ecw,Bonezzi:2023xhn,Borsten:2024dvq} & nontrivial part of the connected tree-level $S$-matrix 
                    \\
                    minimal model of quantum $L_\infty$-algebra~\cite{Pulmann:2016aa,Doubek:2017naz,Jurco:2019yfd,Saemann:2020oyz} & nontrivial part of the connected loop-level $S$-matrix
                    \\
                    further structure on $L_\infty$-algebras (e.g.~$BV^\blacksquare_\infty$-algebra) \cite{Reiterer:2019dys,Borsten:2021hua,Borsten:2022vtg,Borsten:2022ouu,Borsten:2023ned,Borsten:2023reb,Borsten:2023paw,Bonezzi:2023pox,Bonezzi:2023pox,Bonezzi:2024dlv,Armstrong-Williams:2024icu,Bonezzi:2024fhd} & further symmetries (e.g.~colour--kinematics duality)
                    \\
                    \midrule
                    homotopy Maurer--Cartan action for relative $L_\infty$-algebras & perturbative field theory actions including boundary terms 
                    \\
                    minimal model of relative $L_\infty$-algebras & connected tree-level $S$-matrix including trivial part
                    \\
                    \bottomrule
                \end{tabular}
            }
        \end{center}
        \caption{Correspondence between homotopy algebras and quantum field theory physics. This paper focuses on the bottom part of the table.}\label{table:physics-homotopy-dictionary}
    \end{table}

    \paragraph{Related works.}
    The work~\cite{chiaffrino2023holography} is similar in spirit to our discussion of $L_\infty$-algebras for theories with boundary in \Cref{sec:motivating-example} and  holography in \Cref{ssec:ads} and, in part, inspired the present contribution. However, it differs substantially in the technical approach. It would be interesting to understand how these (at least superficially) distinct perspectives are related. 
    
    In particular, rather than recovering the boundary contributions to the BV-action via the pullback from the boundary BFV theory, as would follow from the standard BV--BFV formalism~\cite{Cattaneo_2014}, the bulk and boundary are treated within a single $L_\infty$-algebra in \cite{chiaffrino2023holography}, with modified antifields, products and cyclic structure. In particular, the space  of antifields is enlarged  to include boundary antifields and the differential $\mu_1$ and cyclic structure are adjusted to reproduce the desired bulk action. It is then shown to be possible to homotopy transfer to a boundary theory, which corresponds to a certain \emph{non-minimal} quasi-isomorphic $L_\infty$-algebra. The corresponding homotopy Maurer--Cartan action then computes the boundary action for on-shell field configurations, precisely as one would like for holography. However, the resulting boundary theory is \emph{not} the minimal model of the modified $L_\infty$-algebra, since it has non-trivial differential, see \cite[(3.38)]{chiaffrino2023holography}. It is possible, then, that the minimal model itself is actually trivial with  no physical fields\footnote{For example, for scalar field theory on an oriented compact Riemannian manifold $M$ with boundary $\partial M$, then~\cite[(3.31)]{chiaffrino2023holography} has as the cochain complex 
    \begin{equation}
        \begin{tikzcd}[ampersand replacement=\&,column sep=40pt]
            0\arrow[r]\& \mathcal C^\infty(M)\arrow[r,"{(\Delta-m^2,\partial_N)}"]\&\mathcal C^\infty(M)\oplus\mathcal C^\infty(\partial M)\arrow[r]\& 0,
        \end{tikzcd}
    \end{equation}
    where $\Delta$ is the Beltrami Laplacian, and $\partial_N$ the normal derivative on $\partial M$. The space of physical fields in the minimal model is given by $\ker(\Delta-m^2,\partial_N)$. Since the  Laplacian on a compact Riemannian manifold with boundary $\partial M$ has a non-positive point spectrum for Neumann boundary conditions ($\partial_N\phi=0$), see e.g.~\cite{berard2006spectral}, for $m\not=0$ the cohomology containing the physical states is trivial.} and thus vanishing boundary action. Correspondingly, it is not possible to encode the identity component of the $S$-matrix, or the two-point function of the boundary CFT in a holographic context, in terms of the minimal model. 

    By contrast, we prefer to regard invariance under quasi-isomorphisms as a fundamental guiding principle: quasi-isomorphic $L_\infty$-algebras should be physically equivalent. The physics, e.g.~the $S$-matrix or Witten diagrams, is captured by the relative minimal model, which is unique up to $L_\infty$-isomorphisms, with the boundary  data  encoded in the minimal model morphism $\overset\circ\pi$.

    Our work can be related to the BV--BFV formalism as developed in~\cite{Cattaneo_2014,Mnev_2019,Rejzner_2021}, but the discussion of the minimal model and scattering amplitudes thereof is new to the best of our knowledge. Note that our discussion of holography focuses on the perturbative sector of Witten diagrams and, as such, differs from (and is complementary to) the holography-related discussion in~\cite{Mnev_2019}, which discusses the nonperturbative aspects of AdS\textsubscript3/CFT\textsubscript2 in particular.
    
    The work~\cite{Pulmann:2020omk} discusses the perturbative aspects of AdS\textsubscript3/CFT\textsubscript2 using the Batalin--Vilkovisky formalism; our discussion is related but complementary in that we connect the BV formulation of holography to $L_\infty$-algebras.

    The programme by~\cite{Costello_2021,Paquette:2021cij} to formulate defects and holography using Koszul duality shares many keywords with the current work --- Koszul duality, for example, underlies the definition of $L_\infty$-algebras --- but otherwise differs very much technically; homotopy algebras do not feature heavily in that programme. Nevertheless, it would be interesting to see if the commonalities in concepts could be extended to some sort of concrete connection.

    \section{Relative \texorpdfstring{$L_\infty$}{L∞}-algebras}

    \subsection{A motivating example}\label{sec:motivating-example}

    Let $M$ be an oriented compact Riemannian manifold with metric $g$ and boundary $\partial M$. Consider a scalar field $\phi$ of mass $m$ on $M$ governed by the action
    \begin{equation}\label{eq:spiel-action}
        S\coloneqq-\int_M\operatorname{vol}_M\Big\{\tfrac12(\partial\phi)^2+\tfrac12m^2\phi^2+\tfrac1{3!}\lambda\phi^3\Big\},
    \end{equation}
    where $\operatorname{vol}_M$ is the volume form associated with $g$ and $\lambda$ is the cubic self-interaction coupling constant. The variation of~\eqref{eq:spiel-action} with the Dirichlet boundary condition $\delta\phi|_{\partial M}=0$ yields the desired equation of motion 
    \begin{equation}\label{eq:scalareom}
        (\Delta-m^2)\phi=\lambda\phi^2.
    \end{equation}
    
    To encode this theory in terms of an $L_\infty$-algebra, as loosely described in \Cref{sec:intro}, one first integrates by parts in~\eqref{eq:spiel-action} to arrive at
    \begin{equation}\label{eq:spiel-action-integrated-by-parts} 
        S=\int_M\operatorname{vol}_M\Big\{\tfrac12\phi(\Delta-m^2)\phi-\tfrac1{3!}\lambda\phi^3\Big\}-\tfrac12\int_{\partial M}\operatorname{vol}_{\partial M}\phi\partial_N\phi.
    \end{equation}
    Here, $\Delta$ is the Beltrami Laplacian associated with $g$ and $\partial_N\colon\mathcal C^\infty(M)\rightarrow\mathcal C^\infty(\partial M)$ the normal derivative to $\partial M$ with respect to $g$.\footnote{More precisely, one extends the normal vector field $\partial_N$ on $\partial M\hookrightarrow M$ to some vector field $\tilde V_N$ on $M$ in an arbitrary but smooth fashion. Then, $\partial_N\phi\coloneqq\tilde V_N\phi|_{\partial M}$. Note that this does not depend on the choice of the extension $\tilde V_N$ of $\partial_N$.} Note that the boundary term $\int_{\partial M}\operatorname{vol}_{\partial M}\phi\partial_N\phi$ is precisely what is needed to reproduce the equation of motion \eqref{eq:scalareom} via the variation of \eqref{eq:spiel-action-integrated-by-parts} whilst imposing only the Dirichlet boundary condition on the variation, $\delta\phi|_{\partial M}=0$. Whilst essentially trivial, this toy example captures the generic situation that an action with higher-than-first-order derivatives of a field requires a boundary correction, the most famous instance of which is the Gibbons--Hawking--York term.    
    
    In the present context, the key observation is that the \emph{bulk} term, that is, the first summand in~\eqref{eq:spiel-action-integrated-by-parts}, can be recast as a homotopy Maurer--Cartan action 
    \begin{equation}\label{eq:hMC-action}
        S_\mathrm{hMC}\coloneqq\sum_{n\geq 2}\frac{1}{n!}\langle\phi,\mu_{n-1}(\phi,\dots,\phi)\rangle
    \end{equation} 
    for the $L_\infty$-algebra which has
    \begin{subequations}\label{eq:spiel-bulk-algebra}  
        \begin{equation}
            \LV\coloneqq\underbrace{\mathcal C^\infty(M)}_{\eqqcolon\LV^1}\oplus\underbrace{\mathcal C^\infty(M)}_{\eqqcolon\LV^2} 
        \end{equation} 
        as its underlying graded vector space, where the fields $\phi$ live in $\LV^1$ and the corresponding antifields $\phi^+$ belong to $\LV^2$. The non-trivial $L_\infty$-products\footnote{For simplicity we consider only $\phi^3$ interactions given by $\mu_2$, so that the $L_\infty$-algebra is merely a graded differential Lie algebra, but arbitrary interactions may be included via higher products $\mu_n(\phi_1,\phi_2,\ldots,\phi_n)\coloneqq-\lambda_n\phi_1\phi_2\cdots\phi_n$ of degree $n-2$.}
        \begin{equation}
            \mu_1(\phi)\coloneqq(\Delta-m^2)\phi
            \quad\mbox{and}\quad
            \mu_2(\phi_1,\phi_2)\coloneqq-\lambda\phi_1\phi_2,
        \end{equation}
        of degrees $1$ and $0$, respectively, and the non-degenerate bilinear form
        \begin{equation}\label{eq:spiel-bulk-pairing}
            \langle\phi,\phi'^+\rangle_\LV\coloneqq\int_M\operatorname{vol}_M\phi\phi'^+
        \end{equation} 
    \end{subequations}
    of degree $-3$, which pairs fields with antifields. The latter is not cyclic for $\mu_1$ because of the presence of the boundary $\partial M$, cf.~\cite{chiaffrino2023holography}.

    Can we describe the \emph{boundary} term, that is, the second term in~\eqref{eq:spiel-action-integrated-by-parts}, in a similar language? Naively, we have the boundary-related structures of boundary fields
    \begin{subequations}\label{eq:spiel-naive-boundary-algebra}
        \begin{equation}
            \LV_\partial^\text{naive}\coloneqq\underbrace{\mathcal C^\infty(\partial M)}_{\eqqcolon(\LV_\partial^\text{naive})^1}
        \end{equation}
        and the bulk-to-boundary maps 
        \begin{equation}
            \begin{array}{cccccc}
                \iota^*\,:\, & \LV & \rightarrow & \LV^\text{naive}_\partial,
                \\
                & 
                \begin{pmatrix}
                    \phi\\ \phi^+
                \end{pmatrix}
                & \mapsto & \phi|_{\partial M},
            \end{array}
            \qquad 
            \begin{array}{cccccc}
                \partial_N\,:\, & \LV & \rightarrow & \LV^\text{naive}_\partial,
                \\
                & 
                \begin{pmatrix}
                    \phi\\ \phi^+
                \end{pmatrix}
                & \mapsto & \partial_N\phi,
            \end{array}
        \end{equation}
        given by the pull-back of the natural inclusion $\iota:\LV^\text{naive}_\partial\rightarrow\LV$ and the normal derivative. 
        
        There is also the natural non-degenerate bilinear form on $\LV^\text{naive}_\partial$
        \begin{equation}
            \langle\alpha,\alpha'\rangle_{\LV^\text{naive}_\partial}\coloneqq\int_{\partial M}\operatorname{vol}_{\partial M}\alpha\alpha'
        \end{equation}
    \end{subequations}
    of degree $-2$. With these additional boundary structures  it is straightforward to recast the entire action~\eqref{eq:spiel-action-integrated-by-parts} as
    \begin{equation}
        S=\sum_{n\geq1}\tfrac{1}{n!}\langle\phi,\mu_n(\phi, \ldots,\phi)\rangle_\LV-\tfrac12\big\langle\phi|_{\partial M},\partial_N\phi\big\rangle_{\LV^\text{naive}_\partial}.
    \end{equation}
    The addition of boundary functions $C^\infty(\partial M)$ appears in $L_\infty$-algebra of~\cite{chiaffrino2023holography} for the same reason. However,  in that case it is included in the degree $2$ (antifield) component of the original $L_\infty$-algebra (and is not doubled, as we shall momentarily describe), while we will place it in a relative  boundary $L_\infty$-algebra. Of course, since $C^\infty(\partial M)$ is added as a direct sum in~\cite{chiaffrino2023holography} this is superficially identical (but the degrees are different; here $C^\infty(\partial M)$ is  the space of boundary \emph{fields}).  

    The presence of $\LV_\partial^\text{naive}$ and $\langle-,-\rangle_{\LV^\text{naive}_\partial}$ is reminiscent of the BV--BFV formalism~\cite{Cattaneo_2014,Mnev_2019,Rejzner_2021}.
    The BV--BFV formalism suggests, however, that the boundary should be described by a phase space, not a configuration space --- that is, it should be double the size --- and that
    the two maps $\iota^*,\partial_N\colon\LV\rightarrow\LV_\partial^\text{naive}$ should be bundled up into a single map into a phase space.
    That is, the equations of motion following from the action~\eqref{eq:spiel-action-integrated-by-parts} are of second order so that rather considering the boundary $\partial M$ we should be considering the first-order infinitesimal thickening along the normal bundle. Consequently, we take
    \begin{equation}
        \LV_\partial\coloneqq\underbrace{\mathcal C^\infty(\partial M)\oplus\mathcal C^\infty(\partial M)}_{\eqqcolon\LV_\partial^1}
    \end{equation}
    instead of~\eqref{eq:spiel-naive-boundary-algebra}, and combine $\iota^*$ and $\partial_N$ into a single linear map of degree $0$,
    \begin{equation}
        \begin{array}{ccccc}
            \pi\,:\, & \LV & \rightarrow & \LV_\partial,
            \\
            &
            \begin{pmatrix}
                \phi\\ \phi^+
            \end{pmatrix} 
            & \mapsto &
            \begin{pmatrix}
                \phi|_{\partial M}\\ -\partial_N\phi
            \end{pmatrix}.
        \end{array}
    \end{equation}
    Define a degree $-2$ bilinear form, $\langle-,-\rangle_\partial:\LV_\partial\times\LV_\partial\rightarrow\LV_\partial$, by
    \begin{equation}\label{eq:spiel-bilinear-form}
        \left\langle
            \begin{pmatrix}
                \alpha\\ \beta
            \end{pmatrix},
            \begin{pmatrix}
                \alpha'\\ \beta'
            \end{pmatrix}
        \right\rangle_{\LV_\partial}\coloneqq\int_{\partial M}\operatorname{vol}_{\partial M}\,\alpha\beta'.
    \end{equation}
    It then follows that  
    \begin{equation}\label{eq:spiel-boundary-expr}
        \tfrac12\langle\pi(\phi),\pi(\phi)\rangle_{V_\partial} =-\tfrac12\int_{\partial M}\operatorname{vol}_{\partial M}\phi \partial_N\phi
    \end{equation}
    reproduces the boundary term in the action~\eqref{eq:spiel-action-integrated-by-parts}.  
     
    Note that the term $\int_{\partial M}\operatorname{vol}_{\partial M}\,\alpha\beta'$ is also included in the modified cyclic structure of~\cite{chiaffrino2023holography} for precisely the same reason. However, its definition and interpretation differs from~\eqref{eq:spiel-bilinear-form}. Specifically, in~\cite{chiaffrino2023holography} it is the space of antifields $V^2$ of the $L_\infty$-algebra \eqref{eq:spiel-bulk-algebra} that is enlarged to include a \emph{single} copy of $\mathcal C^\infty(\partial M)$, while in the present case we introduce a second boundary $L_\infty$-algebra with $\LV_\partial=\LV_\partial^{1}\coloneqq{\mathcal C^\infty(\partial M)\oplus\mathcal C^\infty(\partial M)}$ concentrated in the space of \emph{fields}.
    
    We note that the bilinear form~\eqref{eq:spiel-bilinear-form} is degenerate and, in addition, it has no symmetry properties. However, its graded symmetrisation and antisymmetrisation,
    \begin{equation}
        \begin{aligned}
            \langle x,x'\rangle_{\LV_\partial}^\mathrm{sym}\coloneqq\tfrac12\big(\langle x,x'\rangle_{\LV_\partial}+(-1)^{|x||x'|}\langle x',x\rangle_{\LV_\partial}\big),
            \\
            \langle x,x'\rangle_{\LV_\partial}^\mathrm{skew}\coloneqq\tfrac12\big(\langle x,x'\rangle_{\LV_\partial}-(-1)^{|x||x'|}\langle x',x\rangle_{\LV_\partial}\big),
        \end{aligned}
    \end{equation}
    are non-degenerate for all $x\coloneqq(\alpha,\beta)$ and $x'\coloneqq(\alpha',\beta')$,\footnote{Note that, since $\LV_\partial$ is concentrated in degree $1$, we always have $|x|=|x'|=1$.}  and each component separately  plays an important role.
    
    First note, setting $m=0$ for notational simplicity,
    \begin{equation}
    \begin{split}\label{eq:spiel-cyc-correction}
        \langle\phi,\mu_1(\phi')\rangle_\LV& = \tfrac12\int_M\operatorname{vol}_M\phi\Delta\phi'\\
        &=\tfrac12 \int_M\operatorname{vol}_M(\Delta\phi)\phi'+\tfrac12\int_{\partial M}  \phi|_{\partial M}\partial_N\phi'- \tfrac12\int_{\partial M}  \partial_N\phi \phi'|_{\partial M}\\
        &=\langle\phi',\mu_1(\phi)\rangle_\LV+\langle\pi(\phi),\pi(\phi')\rangle_{\LV_\partial}^\mathrm{sym},
        \end{split}
    \end{equation}
    so that  $\langle-,-\rangle_{\LV_\partial}^\mathrm{sym}$ corrects for the failure of the cyclicity of the bilinear form~\eqref{eq:spiel-bulk-pairing} for the differential $\mu_1$ of the bulk $L_\infty$-algebra~\eqref{eq:spiel-bulk-algebra}. 
    
    Note that on setting $\phi=\phi'$ the term $\langle\pi(\phi),\pi(\phi')\rangle_{\LV_\partial}^\mathrm{sym}$ appearing in~\eqref{eq:spiel-cyc-correction} vanishes identically. On the other hand, the graded antisymmetric  component yields
    \begin{equation}
        \langle\pi(\phi),\pi(\phi)\rangle_{\LV_\partial}^\mathrm{skew}=\langle\pi(\phi),\pi(\phi)\rangle_{\LV_\partial},
    \end{equation}
    so that we recover the boundary term in the action~\eqref{eq:spiel-action-integrated-by-parts}.  If one had merely used only $\langle-,-\rangle_{\LV_\partial}^\mathrm{skew}$ at the outset,  the boundary-term of the action would be recovered, but without any correction for cyclicity. 
    
    In summary, $\langle-,-\rangle_{\LV_\partial}^\mathrm{sym}$ corrects for the failure of the cyclicity while $\langle-,-\rangle_{\LV_\partial}^\mathrm{skew}$ corrects for the difference between the Dirichlet action~\eqref{eq:spiel-action} and the homotopy Maurer--Cartan action~\eqref{eq:hMC-action}. 

    It remains to determine the compatibility relation between $\pi$ and the bulk $L_\infty$-algebra products $\mu_1$ and $\mu_2$. It follows that
    \begin{equation}
        \pi\circ\mu_1=0
        \quad\mbox{and}\quad
        \pi\circ\mu_2=0
    \end{equation}
    because of degree reasons since $\pi$ is of degree $0$ and $\LV_\partial$ is concentrated in degree $1$. These identities are consistent with the expected functoriality of $\pi$, i.e.~it should be a chain map, $\pi\circ\mu_1=\mu^\partial_1\circ\pi$ and a Lie algebra homomorphism $\pi(\mu_2(x,x'))=\mu^\partial_2(\pi(x),\pi(x'))$ for all $x$ and $x'$.  

    \subsection{Basic definitions}\label{sec:basicDefinitions}

    In order to formalise the discussion in \Cref{sec:motivating-example} and to set the stage for our later discussion, we now present some basic definitions, motivating them by reference to the key features exposed in \Cref{sec:motivating-example}.

    \paragraph{\mathversion{bold}Cyclic relative $L_\infty$-algebras.}
    We start by recalling the notion of a relative Lie algebra.
    
    \begin{definition}\label[Definition]{def:relative-Lie-algebra}
        A \emph{relative Lie algebra} $(\LV,[-,-],\LV_\partial,[-,-]_\partial,\pi)$ is a pair of Lie algebras $(\LV,[-,-])$ and $(\LV_\partial,[-,-]_\partial)$ together with a morphism $\pi\colon(\LV,[-,-])\rightarrow(\LV_\partial,[-,-]_\partial)$.
        A \emph{morphism of relative Lie algebras} $(\phi,\phi_\partial)\colon(\LV,\LV_\partial)\rightarrow(\LV',\LV'_\partial)$ is a commutative square of Lie algebra homomorphisms
        \begin{equation}
            \begin{tikzcd}
                \LV \dar["\pi"] \rar["\phi"] & \LV' \dar["\pi'"] \\
                \LV_\partial \rar["\phi'"] & \LV'_\partial
            \end{tikzcd}
        \end{equation}
    \end{definition}

    \begin{remark}
        The class of relative Lie algebras forms a two-coloured operad (one colour for $\LV$, another colour for $\LV_\partial$). Following~\cite{markl2002homotopy,markl2001homotopy}, this coloured operad admits a minimal model, i.e.~the corresponding $\infty$-algebra using Koszul duality, which is the following notion.
    \end{remark}

    \begin{definition}\label{def:relative_Linfty}
        A \emph{relative $L_\infty$-algebra} $(\LV,\mu_n,\LV_\partial,\mu^\partial_n, \pi)$ is a pair of $L_\infty$-algebras $ (\LV,\mu_n)$, $(\LV_\partial,\mu^\partial_n)$ together with a morphism $\pi\colon(\LV,\mu_n)\rightarrow(\LV_\partial,\mu^\partial_n)$. Furthermore, a \emph{morphism of relative $L_\infty$-algebras} $(\phi,\phi_\partial)\colon(\LV,\LV_\partial)\rightarrow(\LV',\LV'_\partial)$ is a commutative square of $L_\infty$-algebra morphisms
        \begin{equation}
            \begin{tikzcd}
                \LV \dar["\pi"] \rar["\phi"] & \LV' \dar["\pi'"] \\
                \LV_\partial \rar["\phi'"] & \LV'_\partial
            \end{tikzcd}
        \end{equation}
    \end{definition}
    
    Note that a morphism $\pi$ of $L_\infty$-algebras consists of a family of $n$-linear maps $\pi_n:\LV\times\cdots\times\LV\rightarrow\LV$ of degree $1-n$. It is then clear from the definition that the scalar field theory described in \Cref{sec:motivating-example} yields a relative $L_\infty$-algebra with non-trivial brackets $\mu_2,\mu_1,\mu_1^\partial$ and morphism $\pi$ with only $\pi_1$ non-vanishing. In the context of theories on manifolds with boundary, the general case articulated in \Cref{def:relative_Linfty} allows for arbitrary bulk and boundary interactions, $\mu_n$ and $\mu^\partial_n$.\footnote{Boundary interactions will arise when there are derivative bulk interactions.}

    However, in order to formulate action principles including these interactions, we need the notion of a cyclic structure on a relative $L_\infty$-algebra.

    \begin{definition}\label{def:cyclic_structure}
        A \emph{cyclic structure} on a relative $L_\infty$-algebra 
        \begin{equation}
            (\LV,\mu_n)\xrightarrow\pi(\LV_\partial,\mu^\partial_n)
        \end{equation}
        consists of a non-degenerate graded-symmetric bilinear form
        \begin{subequations}
            \begin{equation}
                \langle-,-\rangle_\LV\colon\LV\times\LV\rightarrow\mathbb R
            \end{equation}
            of degree $-3$ and a bilinear form\footnote{We do neither assume that this bilinear form is non-degenerate nor has any symmetry properties.}
            \begin{equation}
                \langle-,-\rangle_{\LV_\partial}\colon\LV_\partial\times\LV_\partial\rightarrow\mathbb R
            \end{equation}
            of degree $-2$ such that for each $n$ the multilinear maps            \begin{equation}
                \begin{aligned}
                    \LV\times\cdots\times\LV & \rightarrow\mathbb R,
                    \\
                    (x_1,\ldots,x_n) & \mapsto[x_1,\ldots,x_n]_\LV\coloneqq\langle x_1,\mu_{n-1}(x_2,\ldots,x_n)\rangle_\LV
                    \\
                    & \kern4.5cm+\sum_{i+j=n}\langle\pi_i(x_1,\ldots,x_i),\pi_j(x_{i+1},\ldots,x_n)\rangle_{\LV_\partial}
                \end{aligned}
            \end{equation}
            and
            \begin{equation}
                \begin{aligned}
                    \LV_\partial\times\cdots\times\LV_\partial & \rightarrow\mathbb R,
                    \\
                    (x_1,\ldots,x_n) & \mapsto[[x_1,\ldots,x_n]]_{\LV_\partial}\coloneqq\langle x_1,\mu^\partial_{n-1}(x_2,\ldots,x_n)\rangle_{\LV_\partial}
                    \\
                    & \kern4.5cm+(-1)^{|x_1||\mu^\partial_{n-1}(x_2,\ldots,x_n)|}\langle\mu^\partial_{n-1}(x_2,\ldots,x_n),x_1\rangle_{\LV_\partial}
                \end{aligned}
            \end{equation}
            are non-degenerate and also cyclic, 
            \begin{equation}\label{eq:RelCyclic}
                \begin{aligned}
                    [x_1,\ldots,x_n]_\LV &= (-1)^{n-1+(n-1)(|x_1|+|x_n|)+|x_n|\sum_{i=1}^{n-1}|x_i|}[x_n,x_1,\ldots,x_{n-1}]_\LV,
                    \\
                    [[x_1,\ldots,x_n]]_{\LV_\partial} &= (-1)^{n-1+(n-1)(|x_1|+|x_n|)+|x_n|\sum_{i=1}^{n-1}|x_i|}[[x_n,x_1,\ldots,x_{n-1}]]_{\LV_\partial}.
                \end{aligned}
            \end{equation}
        \end{subequations}
    \end{definition}

    \noindent
    Note that the cyclicity will directly impose crossing (Bose) symmetry of scattering amplitudes. 

    \paragraph{Relative homotopy Maurer--Cartan action.}
    Consider a cyclic relative $L_\infty$-algebra
    \begin{equation}\label{eq:relLAlgMCAction}
        (\LV,\mu_n,\langle-,-\rangle_\LV)\xrightarrow\pi(\LV_\partial,\mu^\partial_n,\langle-,-\rangle_{\LV_\partial}).
    \end{equation}
    Motivated by our discussion in \Cref{sec:motivating-example}, we shall refer to $\mathfrak{L}=(\LV,\mu_n,\langle-,-\rangle_\LV)$ as the \emph{bulk} $L_\infty$-algebra and to $\mathfrak{L}_\partial=(\LV_\partial,\mu^\partial_n,\langle-,-\rangle_{\LV_\partial})$ as the \emph{boundary} $L_\infty$-algebra, respectively. To formulate an action principle that encodes all the fields, ghosts, ghosts-for-ghosts, etc.~and also their antifields, it is convenient to consider the degree shift by $[1]$ of~\eqref{eq:relLAlgMCAction}.\footnote{An alternative approach would be to use coordinate functions and the superfield trick~\cite{Cattaneo:0010172}; see \cite[§2.1]{Jurco:2018sby} for a review.} In particular, we have the identifications $\LV[1]\cong[1]\otimes\LV$ and $\LV_\partial[1]\cong[1]\otimes\LV_\partial$ as graded vector spaces, and upon setting
    \begin{subequations}\label{eq:degreeShift}
        \begin{equation}
            \begin{aligned}
                \tilde\mu_n([1]\otimes x_1,\ldots,[1]\otimes x_n) & \coloneqq
                \begin{cases}
                    -\mu_1(x_1) & \mbox{for}\quad n=1
                    \\
                    (-1)^{n+\sum_{i=2}^n\sum_{j=1}^{i-1}|x_j|}\mu_n(x_1,\ldots,x_n) & \mbox{else}
                \end{cases},
                \\
                \tilde\mu_n^\partial([1]\otimes x_1,\ldots,[1]\otimes x_n) & \coloneqq
                \begin{cases}
                    -\mu_1^\partial(x_1) & \mbox{for}\quad n=1
                    \\
                    (-1)^{n+\sum_{i=2}^n\sum_{j=1}^{i-1}|x_j|}\mu_n^\partial(x_1,\ldots,x_n) & \mbox{else}
                \end{cases}
            \end{aligned}
        \end{equation}
        and
        \begin{equation}
            \begin{aligned}
                \langle[1]\otimes x_1,[1]\otimes x_2\rangle_{\LV[1]} & \coloneqq(-1)^{|x_1|}\langle x_1,x_2\rangle_\LV,
                \\
                \langle[1]\otimes x_1,[1]\otimes x_2\rangle_{\LV_\partial[1]} & \coloneqq(-1)^{|x_1|}\langle x_1,x_2\rangle_{\LV_\partial}
            \end{aligned}
        \end{equation}
        as well as
        \begin{equation}
            \tilde\pi([1]\otimes x)\coloneqq[1]\otimes\pi(x),
        \end{equation}
        we obtain the degree-shifted cyclic relative $L_\infty$-algebra
        \begin{equation}\label{eq:degreeShiftedRelativeLInftyAlgebra}
            (\LV[1],\tilde\mu_n,\langle-,-\rangle_{\LV[1]})\xrightarrow{\tilde\pi}(\LV_\partial[1],\tilde\mu^\partial_n,\langle-,-\rangle_{\LV_\partial[1]}).
        \end{equation}
    \end{subequations} 
    
    \begin{definition}\label{def:relative_MC_action}
        Consider the degree-shifted cyclic relative $L_\infty$-algebra~\eqref{eq:degreeShiftedRelativeLInftyAlgebra}. The associated \emph{relative homotopy Maurer--Cartan action} is the expression
        \begin{equation}\label{eq:relative_MC_action}
            \begin{aligned}
                S_\mathrm{rhMC}&\coloneqq\sum_{n\geq2}\frac{1}{n!}[x,\ldots,x]_{\LV[1]}
                \\
                &\phantom{:}=\sum_{n\geq2}\frac{1}{n!}\Big(\langle x,\tilde\mu_{n-1}(x,\ldots,x)\rangle_{\LV[1]}+\sum_{i+j=n}\langle\tilde\pi_i(x,\ldots,x),\tilde\pi_j(x,\ldots,x)\rangle_{\LV_\partial[1]}\Big)
            \end{aligned}
        \end{equation}
        for all $x\in\LV[1]$.
    \end{definition}

    \noindent
    Evidently, the first term in~\eqref{eq:relative_MC_action} is the standard homotopy Maurer--Cartan action~\eqref{eq:hMC-action}, and this will be the only term if the manifold on which the fields live has no boundary. In this case, $\langle-,-\rangle_\LV$ will be  cyclic. Generally, if we have a boundary, the second term will be present. However, if there are no derivative interactions, only $\pi_1$ will be non-trivial.

    \begin{remark}
        Consider a pair of cyclic $L_\infty$-algebras $(\LV,\mu_n,\langle-,-\rangle_\LV)$ and $(\LV',\mu'_n,\langle-,-\rangle_{\LV'})$, that is, both $\langle-,-\rangle_\LV$ and $\langle-,-\rangle_{\LV'}$ are assumed to be cyclic. Recall that a morphism
        \begin{equation}
            (\LV,\mu_n,\langle-,-\rangle_\LV)\xrightarrow\pi(\LV',\mu'_n,\langle-,-\rangle_{\LV'})
        \end{equation}
        is called \emph{cyclic}~\cite{Kajiura:2003ax} (see also~\cite{Jurco:2018sby,Jurco:2019bvp}) provided that
        \begin{subequations}\label{eq:cyclicMorphism}
            \begin{equation}\label{eq:cyclicMorphismA}
                \langle x_1,x_2\rangle_\LV=\langle\pi_1(x_1),\pi_1(x_2)\rangle_{\LV'}
            \end{equation}
            and
            \begin{equation}\label{eq:cyclicMorphismB}
                \sum_{i+j=n}\langle\pi_i(x_1,\ldots,x_i),\pi_j(x_{i+1},\ldots,x_{i+j})\rangle_{\LV'}=0
            \end{equation}
        \end{subequations}
        for all $n\geq3$. Consequently, upon requiring the morphism $\pi$ entering the \Cref{def:relative_Linfty,def:cyclic_structure} to be cyclic in this sense, the second term in the relative homotopy Maurer--Cartan action~\eqref{eq:relative_MC_action} drops out, and we are in the standard situation where we do not have any boundary contributions.\footnote{Note that in this case the condition~\eqref{eq:cyclicMorphismA} means that $\langle\tilde\pi_1(\phi),\tilde\pi_1(\phi)\rangle_{\LV_\partial[1]}=0$ for degree reasons.}
    \end{remark}

    \subsection{Relation to BV--BFV formalism}

    A natural formalism for describing systems with boundaries is given by the Batalin--Vilkovisky--Batalin--Fradkin--Vilkovisky (BV--BFV) formalism~\cite{Cattaneo_2014,Mnev_2019,Rejzner_2021}. It uses the language of graded manifolds, for which see~\cite{Jurco:2018sby,Jurco:2019bvp,Cattaneo:2010re,2105.02534}. The central notion of the BV--BFV formalism is that of a BV--BFV pair.

    \begin{definition}
        A \emph{BV--BFV pair} $(X,Q,\omega)\xrightarrow\Pi(X_\partial,Q_\partial,\omega_\partial)$ consists of a pair of differential graded manifolds $(X,Q)$ and $(X_\partial,Q_\partial)$ together with a morphism $\Pi\colon(X,Q)\rightarrow(X_\partial,Q_\partial)$ and closed two-forms $\omega\in\Omega^2(X)$ and $\omega_\partial\in\Omega^2(X_\partial)$ such that
        \begin{equation}\label{eq:BV-BFV}
            \mathcal L_Q\omega+\Pi^*\omega_\partial=0
            \quad\mbox{and}\quad
            \mathcal L_{Q_\partial}\omega_\partial=0.
        \end{equation}
        Here, $\mathcal L$ denotes the Lie derivative.
    \end{definition}

    Our notion of relative $L_\infty$-algebras relates to the BV--BFV formalism as follows. Firstly, recall that an $L_\infty$-algebra $(\LV,\mu_n)$ can equivalently be described as the differential graded manifold $(\LV[1],Q)$ where the homological vector field $Q$ encodes the products $\mu_n$. In addition, a cyclic structure $\langle-,-\rangle_\LV$ of degree $k$ on $\LV$ is equivalent to a constant (and hence, closed) non-degenerate two-form $\omega\in\Omega^2(\LV[1])$ of degree $k+2$ such that $\mathcal L_Q\omega=0$. We can extend this as follows.

    \begin{proposition}\label{prop:foobar}
        Given a cyclic relative $L_\infty$-algebra
        \begin{equation}
            (\LV,\mu_n,\langle-,-\rangle_\LV)\xrightarrow\pi(\LV_\partial,\mu^\partial_n,\langle-,-\rangle_{\LV_\partial}),
        \end{equation}
        there exists a BV--BFV pair
        \begin{equation}
            (\LV[1],Q,\omega)\xrightarrow\Piit(\LV_\partial[1],Q_\partial,\omega_\partial),
        \end{equation}
        where $(\LV[1],Q)$ and $(\LV_\partial[1],Q_\partial)$ are differential graded manifolds encoding the $L_\infty$-algebras $(\LV,\mu_n)$ and $(\LV_\partial,\mu^\partial_n)$, respectively, and a morphism $\Piit\colon(\LV[1],Q)\rightarrow(\LV_\partial[1],Q_\partial)$ encoding $\pi$ and where $\omega\in\Omega^2(\LV[1])$ and $\omega_\partial\in\Omega^2(\LV_\partial[1])$ are the constant two-forms of degrees $-1$ and $0$, respectively and encoding $\langle-,-\rangle_\LV$ and (the graded-symmetric part of) $\langle-,-\rangle_{\LV_\partial}$, respectively.
    \end{proposition}

    \begin{proof}
        The conditions~\eqref{eq:BV-BFV} are simply the conditions for a cyclic structure in \Cref{def:cyclic_structure}.
    \end{proof}

    \subsection{Minimal model and generalised scattering amplitudes}\label{ssec:minimal_model}

    Since we seek to describe scattering amplitudes, we need to consider the minimal model of a relative $L_\infty$-algebra. Since relative $L_\infty$-algebras are special cases of coloured $\infty$-algebras defined by Koszul duality, they enjoy homotopy transfer theorems; in particular, their minimal models exist. Preserving the cyclic structure requires more work, cf.~\cite{Braun:2013lwa,Braun:2017ikg}.

    \paragraph{Construction of the minimal model.}
    We can construct the minimal model explicitly as follows. Given a cyclic relative $L_\infty$-algebra $(\LV,\mu_n)\xrightarrow\pi(\LV_\partial,\mu_n^\partial)$ and deformation retracts
    \begin{equation}
        \begin{tikzcd}
            (\LV,\mu_1)\ar[loop left, "h"]\rar[shift left, "p_1"]\dar["\pi"] & (H^\bullet(\LV),0) \lar[shift left, "i_1"]
            \\
            (\LV_\partial,\mu_1^\partial)\rar[shift left, "p_1^\partial"]\ar[loop left, "h_\partial"] & (H^\bullet(\LV_\partial),0)\lar[shift left, "i_1^\partial"] 
        \end{tikzcd}
    \end{equation}
    by means of the homological perturbation lemma (see e.g.~\cite{Crainic:0403266,Loday:2012aa,Berglund:0909.3485}) we obtain the $L_\infty$-morphisms
    \begin{subequations}\label{eq:HPL}
        \begin{equation}
            \begin{tikzcd}
                (\LV,\mu_n) \rar[shift left, "p"]\dar["\pi"] & (H^\bullet(\LV),\{\overset\circ\mu{}_{n>1}\}) \lar[shift left, "i"]\dar["\overset\circ\pi"]
                \\
                (\LV_\partial,\mu^\partial_n)\rar[shift left, "p^\partial"] & (H^\bullet(\LV_\partial),\{\overset\circ\mu{}_{n>1}^\partial\})\lar[shift left, "i^\partial"] 
            \end{tikzcd}
        \end{equation}
        with $i$, $p$, $i^\partial$, and $p^\partial$ being $L_\infty$-quasi-isomorphisms and
        \begin{equation}
            \overset\circ\pi\coloneqq p_\partial\circ\pi\circ i.
        \end{equation}
    \end{subequations}
    We call
    \begin{equation}\label{eq:minimalModel}
        (H^\bullet(\LV),\overset\circ\mu{}_{n>1})\xrightarrow{\overset\circ\pi}(H^\bullet(\LV_\partial),\overset\circ\mu{}_{n>1}^\partial)
    \end{equation}
    the \emph{minimal model} of $(\LV,\mu_n)\xrightarrow\pi(\LV_\partial,\mu_n^\partial)$.
    
    Explicitly, we have the following recursion relations expressing the maps $\overset\circ\mu_i$, $\overset\circ\mu{}^\partial_i$, $\overset\circ\pi_i$ of the minimal model in terms of the original maps $\mu_i$, $\mu^\partial_i$, $\pi_i$.
    
    The first two bulk minimal model $L_\infty$-brackets are given by
    \begin{subequations}
        \begin{align}
            \begin{tikzpicture}[
                scale=1.15,
                baseline={([yshift=-0.5ex]current bounding box.center)}
            ]
                \begin{feynhand}
                    \vertex (a) at (-1,-1);
                    \vertex (b) at (1,-1);
                    \vertex (out) at (0,1);
                    \vertex [ringdot] (centre) at (0,0) {$\overset\circ\mu_2$};
                    \propag (a) to (centre);
                    \propag (b) to (centre);
                    \propag (centre) to (out);
                \end{feynhand}
            \end{tikzpicture}
            &=
            \begin{tikzpicture}[
                scale=1.15,
                baseline={([yshift=-0.5ex]current bounding box.center)}
            ]
                \begin{feynhand}
                    \vertex (a) at (-1,-1) {$i$};
                    \vertex (b) at (1,-1) {$i$};
                    \vertex (out) at (0,1) {$p$};
                    \vertex [ringdot] (centre) at (0,0) {$\mu_2$};
                    \propag (a) to (centre);
                    \propag (b) to (centre);
                    \propag (centre) to (out);
                \end{feynhand}
            \end{tikzpicture}
            \\
            \begin{tikzpicture}[
                scale=1.15,
                baseline={([yshift=-0.5ex]current bounding box.center)}
                ]
                \begin{feynhand}
                    \vertex (a) at (-1,-1) ;
                    \vertex (b) at (1,-1) ;
                    \vertex (c) at (0,-1) ;
                    \vertex (out) at (0,1) ;
                    \vertex [ringdot] (centre) at (0,0) {$\overset\circ\mu_3$};
                    \propag (a) to (centre);
                    \propag (b) to (centre);
                    \propag (c) to (centre);
                    \propag (centre) to (out);
                \end{feynhand}
            \end{tikzpicture}
            &=
            \begin{tikzpicture}[
                scale=1.15,
                baseline={([yshift=-0.5ex]current bounding box.center)}
            ]
                \begin{feynhand}
                    \vertex (a) at (-1,-1) {$i$};
                    \vertex (b) at (1,-1) {$i$};
                    \vertex (c) at (0,-1) {$i$};
                    \vertex (out) at (0,1) {$p$};
                    \vertex [ringdot] (centre) at (0,0) {$\mu_3$};
                    \propag (a) to (centre);
                    \propag (b) to (centre);
                    \propag (c) to (centre);
                    \propag (centre) to (out);
                \end{feynhand}
            \end{tikzpicture}
            +
            \begin{tikzpicture}[
                scale=1.15,
                baseline={([yshift=-0.5ex]current bounding box.center)}
            ]
                \begin{feynhand}
                    \vertex (a) at (-1,-1) {$i$};
                    \vertex (b) at (0,-1) {$i$};
                    \vertex (c) at (1,-1) {$i$};
                    \vertex (out) at (0,1) {$p$};
                    \vertex [ringdot] (centre) at (0,0.3) {$\mu_2$};
                    \vertex [ringdot] (centre2) at (-0.3,-0.3) {$\mu_2$};
                    \propag (a) to (centre2);
                    \propag (b) to (centre2);
                    \propag (centre) to [edge label'={$h$}] (centre2);
                    \propag (c) to (centre);
                    \propag (centre) to (out);
                \end{feynhand}
            \end{tikzpicture}
            +
            \begin{tikzpicture}[
                scale=1.15,
                baseline={([yshift=-0.5ex]current bounding box.center)}
            ]
                \begin{feynhand}
                    \vertex (a) at (-1,-1) {$i$};
                    \vertex (b) at (0,-1) {$i$};
                    \vertex (c) at (1,-1) {$i$};
                    \vertex (out) at (0,1) {$p$};
                    \vertex [ringdot] (centre) at (0,0.3) {$\mu_2$};
                    \vertex [ringdot] (centre2) at (0.3,-0.3) {$\mu_2$};
                    \propag (a) to (centre);
                    \propag (b) to (centre2);
                    \propag (centre) to [edge label={$h$}] (centre2);
                    \propag (c) to (centre2);
                    \propag (centre) to (out);
                \end{feynhand}
            \end{tikzpicture}
            +
            \begin{tikzpicture}[
                scale=1.15,
                baseline={([yshift=-0.5ex]current bounding box.center)}
            ]
                \begin{feynhand}
                    \vertex (a) at (-1,-1) {$i$};
                    \vertex (b) at (0,-1) {$i$};
                    \vertex (c) at (1,-1) {$i$};
                    \vertex (out) at (0,1) {$p$};
                    \vertex [ringdot] (centre) at (0.3,0.3) {$\mu_2$};
                    \vertex [ringdot] (centre2) at (-0.3,-0.3) {$\mu_2$};
                    \propag (a) to (centre2);
                    \propag (b) to (centre);
                    \propag (centre) to [edge label'={$h$}](centre2);
                    \propag (c) to (centre2);
                    \propag (centre) to (out);
                \end{feynhand}
            \end{tikzpicture}
        \end{align}
        with the higher bracket given by the obvious generalisation. We recognise these as the usual Feynman diagram expansion, with propagator $h$ and $n$-point vertices $\mu_n$. 
        
        The boundary minimal model $L_\infty$-brackets are analogously given by 
        \begin{align}
            \begin{tikzpicture}[
                scale=1.15,
                baseline={([yshift=-0.5ex]current bounding box.center)}
            ]
                \begin{feynhand}
                    \vertex (a) at (-1,-1) ;
                    \vertex (b) at (1,-1) ;
                    \vertex (out) at (0,1) ;
                    \vertex [ringdot] (centre) at (0,0) {$\overset\circ\mu{}^\partial_2$};
                    \propag (a) to (centre);
                    \propag (b) to (centre);
                    \propag (centre) to (out);
                \end{feynhand}
            \end{tikzpicture}
            &=
            \begin{tikzpicture}[
                scale=1.15,
                baseline={([yshift=-0.5ex]current bounding box.center)}
            ]
                \begin{feynhand}
                    \vertex (a) at (-1,-1) {$i$};
                    \vertex (b) at (1,-1) {$i$};
                    \vertex (out) at (0,1) {$p$};
                    \vertex [ringdot] (centre) at (0,0) {$\mu^\partial_2$};
                    \propag (a) to (centre);
                    \propag (b) to (centre);
                    \propag (centre) to (out);
                \end{feynhand}
            \end{tikzpicture}
            \\
            \begin{tikzpicture}[
                scale=1.15, baseline={([yshift=-0.5ex]current bounding box.center)}
            ]
                \begin{feynhand}
                    \vertex (a) at (-1,-1) ;
                    \vertex (b) at (1,-1) ;
                    \vertex (c) at (0,-1) ;
                    \vertex (out) at (0,1) ;
                    \vertex [ringdot] (centre) at (0,0) {$\overset\circ\mu{}^\partial_3$};
                    \propag (a) to (centre);
                    \propag (b) to (centre);
                    \propag (c) to (centre);
                    \propag (centre) to (out);
                \end{feynhand}
            \end{tikzpicture}
            &=
            \begin{tikzpicture}[
                scale=1.15, baseline={([yshift=-0.5ex]current bounding box.center)}
            ]
                \begin{feynhand}
                    \vertex (a) at (-1,-1) {$i$};
                    \vertex (b) at (1,-1) {$i$};
                    \vertex (c) at (0,-1) {$i$};
                    \vertex (out) at (0,1) {$p$};
                    \vertex [ringdot] (centre) at (0,0) {$\mu^\partial_3$};
                    \propag (a) to (centre);
                    \propag (b) to (centre);
                    \propag (c) to (centre);
                    \propag (centre) to (out);
                \end{feynhand}
            \end{tikzpicture}
            +
            \begin{tikzpicture}[
                scale=1.15,
                baseline={([yshift=-0.5ex]current bounding box.center)}
            ]
                \begin{feynhand}
                    \vertex (a) at (-1,-1) {$i$};
                    \vertex (b) at (0,-1) {$i$};
                    \vertex (c) at (1,-1) {$i$};
                    \vertex (out) at (0,1) {$p$};
                    \vertex [ringdot] (centre) at (0,0.3) {$\mu^\partial_2$};
                    \vertex [ringdot] (centre2) at (-0.3,-0.3) {$\mu^\partial_2$};
                    \propag (a) to (centre2);
                    \propag (b) to (centre2);
                    \propag (centre) to [edge label'={$h$}] (centre2);
                    \propag (c) to (centre);
                    \propag (centre) to (out);
                \end{feynhand}
            \end{tikzpicture}
            +
            \begin{tikzpicture}[
                scale=1.15,
                baseline={([yshift=-0.5ex]current bounding box.center)}
            ]
                \begin{feynhand}
                    \vertex (a) at (-1,-1) {$i$};
                    \vertex (b) at (0,-1) {$i$};
                    \vertex (c) at (1,-1) {$i$};
                    \vertex (out) at (0,1) {$p$};
                    \vertex [ringdot] (centre) at (0,0.3) {$\mu^\partial_2$};
                    \vertex [ringdot] (centre2) at (0.3,-0.3) {$\mu^\partial_2$};
                    \propag (a) to (centre);
                    \propag (b) to (centre2);
                    \propag (centre) to [edge label={$h$}] (centre2);
                    \propag (c) to (centre2);
                    \propag (centre) to (out);
                \end{feynhand}
            \end{tikzpicture}
            +
            \begin{tikzpicture}[
                scale=1.15, baseline={([yshift=-0.5ex]current bounding box.center)}
            ]
                \begin{feynhand}
                    \vertex (a) at (-1,-1) {$i$};
                    \vertex (b) at (0,-1) {$i$};
                    \vertex (c) at (1,-1) {$i$};
                    \vertex (out) at (0,1) {$p$};
                    \vertex [ringdot] (centre) at (0.3,0.3) {$\mu^\partial_2$};
                    \vertex [ringdot] (centre2) at (-0.3,-0.3) {$\mu^\partial_2$};
                    \propag (a) to (centre2);
                    \propag (b) to (centre);
                    \propag (centre) to [edge label'={$h$}](centre2);
                    \propag (c) to (centre2);
                    \propag (centre) to (out);
                \end{feynhand}
            \end{tikzpicture}
        \end{align}        
        Again, there is the obvious generalisation to higher brackets and Feynman diagrammatic interpretation.  
    
        More novel is that the bulk and boundary minimal models are connected by  the minimal model bulk-to-boundary morphisms
        \begin{align}        
            \begin{tikzpicture}[
                scale=1.15,
                baseline={([yshift=-0.5ex]current bounding box.center)}
            ]
                \begin{feynhand}
                    \vertex (a) at (0,-1) ;
                    \vertex (out) at (0,1) ;
                    \vertex [ringdot] (centre) at (0,0) {$\overset\circ\pi_1$};
                    \propag (a) to (centre);
                    \propag (centre) to (out);
                \end{feynhand}
            \end{tikzpicture}
            &=
            \begin{tikzpicture}[
                scale=1.15,
                baseline={([yshift=-0.5ex]current bounding box.center)}
            ]
                \begin{feynhand}
                    \vertex (a) at (0,-1) {$i$};
                    \vertex (out) at (0,1) {$p$};
                    \vertex [ringdot] (centre) at (0,0) {$\pi_1$};
                    \propag (a) to (centre);
                    \propag (centre) to (out);
                \end{feynhand}
            \end{tikzpicture}
            \\
            \begin{tikzpicture}[
                scale=1.15,
                baseline={([yshift=-0.5ex]current bounding box.center)}
            ]
                \begin{feynhand}
                    \vertex (a) at (1,-1) ;
                    \vertex (b) at (-1,-1) ;
                    \vertex (out) at (0,1) ;
                    \vertex [ringdot] (centre) at (0,0) {$\overset\circ\pi_2$};
                    \propag (a) to (centre);
                    \propag (b) to (centre);
                    \propag (centre) to (out);
                \end{feynhand}
            \end{tikzpicture}
            &=
            \begin{tikzpicture}[
                scale=1.15,
                baseline={([yshift=-0.5ex]current bounding box.center)}
            ]
                \begin{feynhand}
                    \vertex (a) at (1,-1) {$i$};
                    \vertex (b) at (-1,-1) {$i$};
                    \vertex (out) at (0,1) {$p$};
                    \vertex [ringdot] (centre) at (0,0) {$\pi_2$};
                    \propag (a) to (centre);
                    \propag (b) to (centre);
                    \propag (centre) to (out);
                \end{feynhand}
            \end{tikzpicture}
            +
            \begin{tikzpicture}[
                scale=1.15,
                baseline={([yshift=-0.5ex]current bounding box.center)}
            ]
                \begin{feynhand}
                    \vertex (a) at (1,-1) {$i$};
                    \vertex (b) at (-1,-1) {$i$};
                    \vertex (out) at (0,1) {$p$};
                    \vertex [ringdot] (centre1) at (0,-0.3) {$\mu_2$};
                    \vertex [ringdot] (centre2) at (0,0.3) {$\pi_1$};
                    \propag (a) to (centre1);
                    \propag (b) to (centre1);
                    \propag (centre1) to [edge label={$h$}] (centre2);
                    \propag (centre2) to (out);
                \end{feynhand}
            \end{tikzpicture}
            +
            \begin{tikzpicture}[
                scale=1.15,
                baseline={([yshift=-0.5ex]current bounding box.center)}
            ]
                \begin{feynhand}
                    \vertex (a) at (-1,-1) {$i$};
                    \vertex (b) at (1,-1) {$i$};
                    \vertex (out) at (0,1) {$p$};
                    \vertex [ringdot] (centre1a) at (-0.3,-0.3) {$\pi_1$};
                    \vertex [ringdot] (centre1b) at (0.3,-0.3) {$\pi_1$};
                    \vertex [ringdot] (centre2) at (0,0.3) {$\mu^\partial_2$};
                    \propag (a) to (centre1a);
                    \propag (b) to (centre1b);
                    \propag (centre1a) to [edge label={$h$}] (centre2);
                    \propag (centre1b) to [edge label'={$h$}] (centre2);
                    \propag (centre2) to (out);
                \end{feynhand}
            \end{tikzpicture}
        \end{align}
        \begin{align}
            \begin{tikzpicture}[
                scale=1.15,
                baseline={([yshift=-0.5ex]current bounding box.center)}
            ]
                \begin{feynhand}
                    \vertex (a) at (-1,-1) ;
                    \vertex (b) at (1,-1) ;
                    \vertex (c) at (0,-1) ;
                    \vertex (out) at (0,1) ;
                    \vertex [ringdot] (centre) at (0,0) {$\overset\circ\pi_3$};
                    \propag (a) to (centre);
                    \propag (b) to (centre);
                    \propag (c) to (centre);
                    \propag (centre) to (out);
                \end{feynhand}
            \end{tikzpicture}
            &=
            \begin{tikzpicture}[
                scale=1.15,
                baseline={([yshift=-0.5ex]current bounding box.center)}
            ]
                \begin{feynhand}
                    \vertex (a) at (-1,-1) {$i$};
                    \vertex (b) at (1,-1) {$i$};
                    \vertex (c) at (0,-1) {$i$};
                    \vertex (out) at (0,1) {$p$};
                    \vertex [ringdot] (centre) at (0,0) {$\pi_3$};
                    \propag (a) to (centre);
                    \propag (b) to (centre);
                    \propag (c) to (centre);
                    \propag (centre) to (out);
                \end{feynhand}
            \end{tikzpicture}
            +
            \begin{tikzpicture}[
                scale=1.15,
                baseline={([yshift=-0.5ex]current bounding box.center)}
            ]
                \begin{feynhand}
                    \vertex (a) at (-1,-1) {$i$};
                    \vertex (b) at (0,-1) {$i$};
                    \vertex (c) at (1,-1) {$i$};
                    \vertex (out) at (0,1) {$p$};
                    \vertex [ringdot] (centre) at (0,0.3) {$\pi_2$};
                    \vertex [ringdot] (centre2) at (-0.3,-0.3) {$\mu_2$};
                    \propag (a) to (centre2);
                    \propag (b) to (centre2);
                    \propag (centre) to [edge label'={$h$}] (centre2);
                    \propag (c) to (centre);
                    \propag (centre) to (out);
                \end{feynhand}
            \end{tikzpicture}
            +
            \begin{tikzpicture}[
                scale=1.15,
                baseline={([yshift=-0.5ex]current bounding box.center)}
            ]
                \begin{feynhand}
                    \vertex (a) at (-1,-1) {$i$};
                    \vertex (b) at (0,-1) {$i$};
                    \vertex (c) at (1,-1) {$i$};
                    \vertex (out) at (0,1) {$p$};
                    \vertex [ringdot] (centre) at (0,0.3) {$\pi_2$};
                    \vertex [ringdot] (centre2) at (0.3,-0.3) {$\mu_2$};
                    \propag (a) to (centre);
                    \propag (b) to (centre2);
                    \propag (centre) to [edge label={$h$}] (centre2);
                    \propag (c) to (centre2);
                    \propag (centre) to (out);
                \end{feynhand}
            \end{tikzpicture}
            +
            \begin{tikzpicture}[
                scale=1.15,
                baseline={([yshift=-0.5ex]current bounding box.center)}
            ]
                \begin{feynhand}
                    \vertex (a) at (-1,-1) {$i$};
                    \vertex (b) at (0,-1) {$i$};
                    \vertex (c) at (1,-1) {$i$};
                    \vertex (out) at (0,1) {$p$};
                    \vertex [ringdot] (centre) at (0.3,0.3) {$\pi_2$};
                    \vertex [ringdot] (centre2) at (-0.3,-0.3) {$\mu_2$};
                    \propag (a) to (centre2);
                    \propag (b) to (centre);
                    \propag (centre) to [edge label'={$h$}](centre2);
                    \propag (c) to (centre2);
                    \propag (centre) to (out);
                \end{feynhand}
            \end{tikzpicture}
            \\
            &\qquad+
            \begin{tikzpicture}[
                scale=1.15,
                baseline={([yshift=-0.5ex]current bounding box.center)}
            ]
                \begin{feynhand}
                    \vertex (a) at (-1,-1) {$i$};
                    \vertex (b) at (0,-1) {$i$};
                    \vertex (c) at (1,-1) {$i$};
                    \vertex (out) at (0,1) {$p$};
                    \vertex [ringdot] (centre) at (0,0.3) {$\mu^\partial_2$};
                    \vertex [ringdot] (centre2) at (-0.3,-0.3) {$\pi_2$};
                    \propag (a) to (centre2);
                    \propag (b) to (centre2);
                    \propag (centre) to [edge label'={$h$}] (centre2);
                    \propag (c) to (centre);
                    \propag (centre) to (out);
                \end{feynhand}
            \end{tikzpicture}
            +
            \begin{tikzpicture}[
                scale=1.15,
                baseline={([yshift=-0.5ex]current bounding box.center)}
            ]
                \begin{feynhand}
                    \vertex (a) at (-1,-1) {$i$};
                    \vertex (b) at (0,-1) {$i$};
                    \vertex (c) at (1,-1) {$i$};
                    \vertex (out) at (0,1) {$p$};
                    \vertex [ringdot] (centre) at (0,0.3) {$\mu^\partial_2$};
                    \vertex [ringdot] (centre2) at (0.3,-0.3) {$\pi_2$};
                    \propag (a) to (centre);
                    \propag (b) to (centre2);
                    \propag (centre) to [edge label={$h$}] (centre2);
                    \propag (c) to (centre2);
                    \propag (centre) to (out);
                \end{feynhand}
            \end{tikzpicture}
            +
            \begin{tikzpicture}[
                scale=1.15,
                baseline={([yshift=-0.5ex]current bounding box.center)}
            ]
                \begin{feynhand}
                    \vertex (a) at (-1,-1) {$i$};
                    \vertex (b) at (0,-1) {$i$};
                    \vertex (c) at (1,-1) {$i$};
                    \vertex (out) at (0,1) {$p$};
                    \vertex [ringdot] (centre) at (0.3,0.3) {$\mu^\partial_2$};
                    \vertex [ringdot] (centre2) at (-0.3,-0.3) {$\pi_2$};
                    \propag (a) to (centre2);
                    \propag (b) to (centre);
                    \propag (centre) to [edge label'={$h$}](centre2);
                    \propag (c) to (centre2);
                    \propag (centre) to (out);
                \end{feynhand}
            \end{tikzpicture}\notag
            \\
            &\qquad+
            \begin{tikzpicture}[
                scale=1.15,
                baseline={([yshift=-0.5ex]current bounding box.center)}
            ]
                \begin{feynhand}
                    \vertex (a) at (-1,-1) {$i$};
                    \vertex (b) at (1,-1) {$i$};
                    \vertex (c) at (0,-1) {$i$};
                    \vertex (out) at (0,1) {$p$};
                    \vertex [ringdot] (centre1) at (0,-0.3) {$\mu_3$};
                    \vertex [ringdot] (centre2) at (0,0.3) {$\pi_1$};
                    \propag (a) to (centre1);
                    \propag (b) to (centre1);
                    \propag (c) to (centre1);
                    \propag (centre1) to [edge label={$h$}](centre2);
                    \propag (centre2) to (out);
                \end{feynhand}
            \end{tikzpicture}
            +
            \begin{tikzpicture}[
                scale=1.15,
                baseline={([yshift=-0.5ex]current bounding box.center)}
            ]
                \begin{feynhand}
                    \vertex (a) at (-1,-1) {$i$};
                    \vertex (b) at (0,-1) {$i$};
                    \vertex (c) at (1,-1) {$i$};
                    \vertex (out) at (0,1) {$p$};
                    \vertex [ringdot] (centre1a) at (-0.5,-0.3) {$\pi_1$};
                    \vertex [ringdot] (centre1b) at (0,-0.3) {$\pi_1$};
                    \vertex [ringdot] (centre1c) at (0.5,-0.3) {$\pi_1$};
                    \vertex [ringdot] (centre2) at (0,0.3) {$\mu^\partial_3$};
                    \propag (a) to (centre1a);
                    \propag (b) to (centre1b);
                    \propag (c) to (centre1c);
                    \propag (centre1a) to [edge label={$h$}] (centre2);
                    \propag (centre1b) to [edge label={$h$}] (centre2);
                    \propag (centre1c) to [edge label'={$h$}] (centre2);
                    \propag (centre2) to (out);
                \end{feynhand}
            \end{tikzpicture}.\notag
        \end{align}
    \end{subequations}
    
    \paragraph{Generalised scattering amplitudes.}
    Using the minimal model, we can now define the generalised scattering amplitudes as follows. In a quantum field theory, the full $S$-matrix,
    \begin{equation}
        S_{2\to2} =
        \underbrace{
            \begin{tikzpicture}[
                baseline={([yshift=-0.5ex]current bounding box.center)},
                scale=0.4
            ]
                \begin{feynhand}
                    \vertex (a) at (-1,-1);
                    \vertex (b) at (1,-1);
                    \vertex (c) at (-1,1);
                    \vertex (d) at (1,1);
                    \vertex (o) at (0,0);
                    \propag (a) to (c);
                    \propag (b) to (d);
                \end{feynhand}
            \end{tikzpicture}
            ~~~
            +
            ~~~
            \begin{tikzpicture}[
                baseline={([yshift=-0.5ex]current bounding box.center)},
                scale=0.4
            ]
                \begin{feynhand}
                    \vertex (a) at (-1,-1);
                    \vertex (b) at (1,-1);
                    \vertex (c) at (-1,1);
                    \vertex (d) at (1,1);
                    \vertex (o) at (0,0);
                    \propag (a) to (d);
                    \propag (b) to (c);
                \end{feynhand}
            \end{tikzpicture}
        }_{\text{identity part}}
        ~~~
        +
        ~~~
        \begin{tikzpicture}[
            baseline={([yshift=-0.5ex]current bounding box.center)},
            scale=0.4
        ]
            \begin{feynhand}
                \vertex (a) at (-1,-1);
                \vertex (b) at (1,-1);
                \vertex (c) at (-1,1);
                \vertex (d) at (1,1);
                \vertex (o) at (0,0);
                \vertex [dot] (o1) at (0,-0.5) {};
                \vertex [dot] (o2) at (0,0.5) {};
                \propag (a) to (o1) to (b);
                \propag (o1) to (o2);
                \propag (c) to (o2) to (d);
            \end{feynhand}
        \end{tikzpicture}
        ~~~
        +
        ~~~
        \begin{tikzpicture}[
            baseline={([yshift=-0.5ex]current bounding box.center)},
            scale=0.4
        ]
            \begin{feynhand}
                \vertex (a) at (-1,-1);
                \vertex (b) at (1,-1);
                \vertex (c) at (-1,1);
                \vertex (d) at (1,1);
                \vertex (o) at (0,0);
                \vertex [dot] (o1) at (-0.5,0) {};
                \vertex [dot] (o2) at (0.5,0) {};
                \propag (a) to (o1) to (c);
                \propag (o1) to (o2);
                \propag (b) to (o2) to (d);
            \end{feynhand}
        \end{tikzpicture}
        ~~~
        +
        ~~~
        \begin{tikzpicture}[
            baseline={([yshift=-0.5ex]current bounding box.center)},
            scale=0.4
        ]
            \begin{feynhand}
                \vertex (a) at (-1,-1);
                \vertex (b) at (1,-1);
                \vertex (c) at (-1,1);
                \vertex (d) at (1,1);
                \vertex (o) at (0,0);
                \vertex [dot] (o1) at (-0.5,0) {};
                \vertex [dot] (o2) at (0.5,0) {};
                \propag (a) to (o1) to (d);
                \propag (o1) to (o2);
                \propag (c) to (o2) to (b);
            \end{feynhand}
        \end{tikzpicture}
        ~~~
        +
        ~~~
        \begin{tikzpicture}[
            baseline={([yshift=-0.5ex]current bounding box.center)},
            scale=0.4
        ]
            \begin{feynhand}
                \vertex (a) at (-1,-1);
                \vertex (b) at (1,-1);
                \vertex (c) at (-1,1);
                \vertex (d) at (1,1);
                \vertex [dot] (o) at (0,0) {};
                \propag (a) to (d);
                \propag (b) to (c);
            \end{feynhand}
        \end{tikzpicture}
        ~~~
        +
        ~~~
        \cdots,
    \end{equation}
    comes from exponentiating the connected (Wick) diagrams,
    \begin{equation}
        \left\{
            \quad
            \begin{tikzpicture}[
                baseline={([yshift=-0.5ex]current bounding box.center)},
                scale=0.4
            ]
                \begin{feynhand}
                    \vertex (a) at (0,-1);
                    \vertex (b) at (0,1);
                    \propag (a) to (b);
                \end{feynhand}
            \end{tikzpicture}
            \quad,\quad
            \begin{tikzpicture}[
                baseline={([yshift=-0.5ex]current bounding box.center)},
                scale=0.4
            ]
                \begin{feynhand}
                    \vertex (a) at (-1,-1);
                    \vertex (b) at (1,-1);
                    \vertex (c) at (-1,1);
                    \vertex (d) at (1,1);
                    \vertex (o) at (0,0);
                    \vertex [dot] (o1) at (0,-0.5) {};
                    \vertex [dot] (o2) at (0,0.5) {};
                    \propag (a) to (o1) to (b);
                    \propag (o1) to (o2);
                    \propag (c) to (o2) to (d);
                \end{feynhand}
            \end{tikzpicture}
            ,\quad
            \begin{tikzpicture}[
                baseline={([yshift=-0.5ex]current bounding box.center)},
                scale=0.4
            ]
                \begin{feynhand}
                    \vertex (a) at (-1,-1);
                    \vertex (b) at (1,-1);
                    \vertex (c) at (-1,1);
                    \vertex (d) at (1,1);
                    \vertex (o) at (0,0);
                    \vertex [dot] (o1) at (-0.5,0) {};
                    \vertex [dot] (o2) at (0.5,0) {};
                    \propag (a) to (o1) to (c);
                    \propag (o1) to (o2);
                    \propag (b) to (o2) to (d);
                \end{feynhand}
            \end{tikzpicture}
            ,\quad
            \begin{tikzpicture}[
                baseline={([yshift=-0.5ex]current bounding box.center)},
                scale=0.4
            ]
                \begin{feynhand}
                    \vertex (a) at (-1,-1);
                    \vertex (b) at (1,-1);
                    \vertex (c) at (-1,1);
                    \vertex (d) at (1,1);
                    \vertex (o) at (0,0);
                    \vertex [dot] (o1) at (-0.5,0) {};
                    \vertex [dot] (o2) at (0.5,0) {};
                    \propag (a) to (o1) to (d);
                    \propag (o1) to (o2);
                    \propag (c) to (o2) to (b);
                \end{feynhand}
            \end{tikzpicture}
            ,\quad
            \begin{tikzpicture}[
                baseline={([yshift=-0.5ex]current bounding box.center)},
                scale=0.4
            ]
                \begin{feynhand}
                    \vertex (a) at (-1,-1);
                    \vertex (b) at (1,-1);
                    \vertex (c) at (-1,1);
                    \vertex (d) at (1,1);
                    \vertex [dot] (o) at (0,0) {};
                    \propag (a) to (d);
                    \propag (b) to (c);
                \end{feynhand}
            \end{tikzpicture}
            ,\quad\ldots
        \right\}.
    \end{equation}
    Of the diagrams that comprise $S_\text{conn}$, the `trivial' diagram `$|$' is special. The minimal model of an $L_\infty$-algebra contains the data for all connected diagrams \emph{except} for the trivial diagram, but the trivial diagram is crucial in reproducing the full $S$-matrix. And, in fact, the `trivial' diagram is not so trivial in general in curved geometries such as anti-de~Sitter space.

    \begin{definition}\label{def:generalized_amplitude}
        Consider the minimal model~\eqref{eq:minimalModel} of a relative cyclic $L_\infty$-algebra as well its degree shift constructed by means of~\eqref{eq:degreeShift}. At the tree level, the associated \emph{ generalised connected $n$-point scattering amplitude} is given by the expression
        \begin{equation}\label{eq:generalized_amplitude}
            \begin{aligned}
                [\phi_1,\ldots,\phi_n]_{H^\bullet(\LV)[1]} & =\langle\phi_1,\tilde{\overset\circ\mu}_{n-1}(\phi_2,\ldots,\phi_n)\rangle_{H^\bullet(\LV)}
                \\
                & \kern2cm+\sum_{i+j=n}\langle\tilde{\overset\circ\pi}_i(\phi_1,\ldots,\phi_i),\tilde{\overset\circ\pi}_j(\phi_{i+1},\ldots,\phi_n)\rangle_{H^\bullet(\LV_\partial)}
            \end{aligned}
        \end{equation}
        for $\phi_{1,\ldots,n}\in H^\bullet(\LV)[1]$.\footnote{Put differently, the generalised scattering amplitudes follow from the polarisation of the relative homotopy Maurer--Cartan action~\eqref{eq:relative_MC_action} for the degree-shifted minimal model of a relative cyclic $L_\infty$-algebra.}
    \end{definition}

    \noindent
    Note that, whilst $\langle-,-\rangle_{H^\bullet(\LV)}$ need not be cyclic with respect to $\overset\circ\mu_i$, any failure of cyclicity is compensated by the antisymmetric part of $\langle-,-\rangle_{H^\bullet(\LV_\partial)}$ according to \Cref{def:cyclic_structure}, such that the $n$-point connected scattering amplitudes are guaranteed to respect crossing (Bose) symmetry. Note also that, for the two-point scattering amplitude, whilst the first term in~\eqref{eq:generalized_amplitude} vanishes since $\overset\circ\mu_1=0$ by construction, the second term will, in general, not be zero, so that we recover the trivial piece of the $S$-matrix.

    The full (rather than connected) tree-level $S$-matrix can be easily defined in terms of the connected tree-level $S$-matrix. For the full quantum $S$-matrix, there is an evident generalisation to the loop case following~\cite{Jurco:2019yfd,Saemann:2020oyz}, where one should consider a loop relative $L_\infty$-algebra.

    \paragraph{Boundary contributions to higher-point scattering amplitudes.}
    In the above, suppose that $\pi$ is strict, that is, $\pi_{n>1}=0$ and that
    \begin{equation}
        \pi_1\circ h=0
    \end{equation}
    for $h$ the contracting homotopy. According to explicit formulas for homotopy transfer, the higher-order components $i_{n>1}$ of the $L_\infty$-quasi-isomorphism $i$ in~\eqref{eq:HPL} satisfy
    \begin{equation}
        \operatorname{im}(i_n)\subseteq\operatorname{im}(h).
    \end{equation}
    Then it follows that the boundary terms in~\eqref{eq:generalized_amplitude} are trivial except for two-point scattering amplitudes. Therefore, in this case, we only have boundary corrections to the two-point scattering amplitudes but not to higher-point ones, consistent with the fact that the usual $L_\infty$-algebra formalism (ignoring boundary contributions) yield the correct connected tree-level scattering amplitudes for $n>2$ points.

    \section{Examples}

    Let us now discuss applications of the above setup. We start off considering the simplest examples of scalar field theory, before moving on to Chern--Simons theory and Yang--Mills theory. 

    \subsection{Scalar field theory on a manifold-with-boundary}\label{ssec:scalar_field}

    Having setting out the definitions in \Cref{sec:basicDefinitions}, we can now  return to the   motivating example from \Cref{sec:motivating-example},  first concisely summarising the general setup for compact Riemannian manifold with boundary and then presenting the corresponding relative minimal model. 
 
    We then consider the more subtle cases of Euclidean flat  and anti-de Sitter space, where the boundaries are asymptotic. In the latter case, we shall reproduce the well-known results from the AdS/CFT literature. See also~\cite{chiaffrino2023holography} for  analogous results, obtained in a conceptually similar manner (although, since our underlying framework is different, the technical details remain distinct).  
 
    \subsubsection{Compact Riemannian manifolds-with-boundary}

    Consider a scalar field on scalar field $\phi$ of mass $m$ on an oriented compact Riemannian manifold $(M,g)$ with boundary $\partial M$ and cubic interactions.
 
    \paragraph{\mathversion{bold}Relative $L_\infty$-algebra.} 
    The associated cochain complex of the relative $L_\infty$-algebra is\footnote{This complex differs from the analogous structures in~\cite{chiaffrino2023holography} (see eqq.~(3.31) and (3.38) therein) in two regards. Firstly, the `bulk' and `boundary' (anti)fields are contained in single relative $L_\infty$-algebra, as opposed to belonging to  quasi-isomorphic, but distinct, $L_\infty$-algebras related by homotopy retract data. Indeed, we have a homotopy transfer from the full relative bulk$+$boundary relative $L_\infty$-algebra to a minimal bulk$+$boundary relative $L_\infty$-algebra. Secondly, the boundary fields are doubled and the space of boundary antifields is trivial.}
    \begin{subequations}
        \begin{equation}
            \begin{tikzcd}
                0 \rar & C^\infty(M)\rar["\mu_1"]\dar["\pi_1"] & C^\infty(M)\rar\dar["0"] & 0 
                \\
                0 \rar & C^\infty(\partial M)\oplus C^\infty(\partial M)\rar & 0
            \end{tikzcd}
        \end{equation}
        where the first row represents $\LV$ and the second row $\LV_\partial$ in the notation of \Cref{sec:basicDefinitions}. The non-trivial brackets and components of $\pi$ are
        \begin{equation}
            \begin{gathered}
                \mu_1(\phi)\coloneqq(\Delta-m^2)\phi,
                \quad
                \mu_2(\phi,\phi')\coloneqq-\lambda\phi\phi',
                \\
                \pi_1(\phi)\coloneqq
                \begin{pmatrix}
                    \phi|_{\partial M}\\ -\partial_N\phi
                \end{pmatrix},
            \end{gathered}
        \end{equation}
        where $\lambda$ is the coupling constant. As in \Cref{sec:motivating-example}, $\partial_N$ is (the differential operator associated with) the vector field normal to $\partial M$, and $\phi|_{\partial M}$ is the restriction of $\phi$ to $\partial M$. Note that the doubled function space in $\LV_\partial$ can be seen as the space of functions on a first-order neighbourhood of $\partial M$. Furthermore, we set
        \begin{equation}
            \langle\phi,\phi'^+\rangle_\LV\coloneqq\int_M\operatorname{vol}_M\phi\phi'^+\eqqcolon\langle\phi'^+,\phi\rangle_\LV
            \quad\mbox{and}\quad
            \left\langle
            \begin{pmatrix}
                \alpha\\ \beta
            \end{pmatrix},
            \begin{pmatrix}
                \alpha'\\ \beta'
            \end{pmatrix}
            \right\rangle_{\LV_\partial}
            \coloneqq\int_{\partial M}\operatorname{vol}_{\partial M}\alpha\beta'.
        \end{equation}
    \end{subequations}

    The relative homotopy Maurer--Cartan action~\eqref{eq:relative_MC_action} for this relative $L_\infty$-algebra is
    \begin{equation}
        \begin{aligned}
            S&=\int_M\operatorname{vol}_M\Big\{\tfrac12\phi(\Delta-m^2)\phi-\tfrac1{3!}\lambda\phi^3\Big\}-\tfrac12\int_{\partial M}\operatorname{vol}_{\partial M}\phi \partial_N\phi
            \\
            &=-\int_M\operatorname{vol}_M\Big\{\tfrac12(\partial\phi)^2+\tfrac12m^2\phi^2+\tfrac1{3!}\lambda\phi^3\Big\}.
        \end{aligned}
    \end{equation}

    \paragraph{Relative minimal model.}    
    The cochain complex of the minimal model of the bulk theory is 
    \begin{equation}
        \begin{tikzcd}
            0\rar & \ker(\Delta-m^2)\rar["{\overset\circ\mu_1}=0"]& \operatorname{coker}(\Delta-m^2)\rar& 0.
        \end{tikzcd}
    \end{equation}
    The homotopy $h$ is given by a choice of Green's function\footnote{If one imposes Dirichlet boundary conditions the Green's function is unique.} $G_m$ for $\Delta-m^2$ (understood to carry $L_\infty$-degree $-1$). Then, the trivial embeddings of $\ker(\Delta-m^2)$ into $C^\infty(M)$ and the obvious projections $\textrm{id}-\mu_1\circ G_m:C^\infty(M)\rightarrow\ker(\Delta-m^2)$ and $\textrm{id}-G_m\circ\mu_1:C^\infty(M)\rightarrow\operatorname{coker}(\Delta-m^2)$ provide the homotopy retract data. The higher products $\overset\circ\mu_{n>1}$ then follow from the homological perturbation lemma, using the formulae of \Cref{ssec:minimal_model}. See also~\cite{Macrelli:2019afx} for more details.

    The boundary minimal model is trivially given by the boundary $L_\infty$-algebra itself, so the  cochain complex for the relative minimal model is given by
    \begin{equation}
        \begin{tikzcd}
            0\rar & \ker(\Delta-m^2)\rar["{\overset\circ\mu_1}=0"]\dar["{\overset\circ\pi}_1=\pi_1\big|_{\ker(\Delta-m^2)}"] & \operatorname{coker}(\Delta-m^2)\rar\dar["\overset\circ\pi_1=0"] & 0
            \\
            0\rar & C^\infty(\partial M)\oplus C^\infty(\partial M)\rar & 0
        \end{tikzcd}
    \end{equation}
    where the first row represents $H^\bullet(\LV)$ and the second row $H^\bullet(\LV_\partial)$, and there are no boundary higher products $\overset\circ\mu{}^\partial_n$ (since there are no derivative interactions in the bulk). The minimal model then encodes Feynman-diagram-like perturbation theory for solutions to the classical equations of motion.

    \subsubsection{Flat and anti-de~Sitter spaces}\label{ssec:ads}

    In this section, we discuss field theories on manifolds with a boundary `at infinity', namely on anti-de~Sitter space and flat space, to reproduce conformal-field-theory (CFT) correlators (computed via Witten diagrams) and $S$-matrices (computed via Feynman diagrams), respectively. Our discussion for the flat-space $S$-matrix will be inspired by the AdS/CFT correspondence; for a related approach see~\cite{kim2023smatrix}. In both cases, it is more convenient to motivate the construction in Euclidean (rather than Lorentzian) signature, although one can always Wick-rotate to Lorentzian structure afterwards. The idea that the trivial part of the $S$-matrix arises from boundary terms also appears in~\cite{caronhuot2023measured}.

    \paragraph{Euclidean anti-de~Sitter space.}
    Consider the Poincar{\'e} patch of Euclidean anti-de~Sitter space (a.k.a.~hyperbolic space) ${\rm AdS}_{d+1}$. The underlying manifold is ${\rm AdS}_{d+1}\coloneqq\{(z,\vec y)|z\in\mathbb R_{>0},\vec y\in\mathbb R^d\}$, and the Riemannian metric $g$ is given by 
    \begin{equation}
        g_\mathrm{AdS}\coloneqq\frac1{z^2}\left(\mathrm dz\otimes\mathrm dz+\sum_{i=1}^d\mathrm dy^i\otimes\mathrm dy^i\right).
    \end{equation}
    The conformal boundary $S^d\cong \mathbb{R}^d\cup \infty$ lies at $z=0$ and $z=\infty$.

    \paragraph{Euclidean flat space.}
    We treat $(d+1)$-dimensional Euclidean flat space similarly to the hyperbolic space: the underlying manifold is $\{(z,\vec y)|z\in\mathbb R_{>0},\vec y\in\mathbb R^d\}$, and the Riemannian metric $g$ is given by
    \begin{equation}
        g_\mathrm{E}\coloneqq\frac1{z^2}\mathrm dz\otimes\mathrm dz+\sum_{i=1}^d\mathrm dy^i\otimes\mathrm dy^i.
    \end{equation}
    Using the coordinate transformation $z\eqqcolon\mathrm e^t$, then $(t,\vec y)$ are the usual Cartesian coordinates for $\mathbb R^{1+d}$. Since this is in Euclidean signature, there is only one component of the conformal boundary (similar to hyperbolic space), which lies at $z=0$, corresponding to the far past $t\to-\infty$; there is no corresponding boundary component for the far future $t\to+\infty$. Effectively, we are using crossing symmetry to replace outgoing legs of positive energy with incoming legs of negative energy so that all external legs come in from the far past.

    To treat the hyperbolic and flat cases uniformly, we write
    \begin{equation}
        g\coloneqq\frac1{z^2}\mathrm dz\otimes\mathrm dz+\frac1{z^{2n}}\sum_{i=1}^d\mathrm dy^i\otimes\mathrm dy^i
    \end{equation}
    with $n=0$ for flat space and $n=1$ for hyperbolic space, and set $z\eqqcolon\mathrm e^t$.

    \paragraph{Bulk function space.}
    Let $C$ be the space of smooth functions $\phi\colon\mathbb R_{>0}\times\mathbb R^d\rightarrow\mathbb R$ such that there exists a function $\phi_\text{interior}\colon\mathbb R_{>0}\times\mathbb R^d\rightarrow\mathbb R$ that decays superpolynomially as $z\to0$ or, equivalently, superexponentially as $t\to-\infty$ (i.e.~$\lim_{z\to0}z^{-\alpha}\phi_\mathrm{interior}(z,\vec y)=0$ or $\lim_{t\to-\infty}\mathrm e^{-\alpha t}\phi_\mathrm{interior}(\mathrm e^t,\vec y)=0$ at every $\vec y\in\mathbb R^d$ and $\alpha\in\mathbb R$) and the difference $\phi_\mathrm{pw}\coloneqq\phi-\phi_\text{interior}$ is a countable\footnote{We allow countable, rather than finite, sums so as to include solutions to the Helmholtz equation on hyperbolic space.} sum of functions that depend homogeneously on $z$ as 
    \begin{equation}\label{eq:phi_decomposition}
        \phi(z,\vec y)=\phi_\text{interior}(z,\vec y)+\sum_{i=1}^\infty z^{\alpha_i}\phi_{\alpha_i}(\vec y),
    \end{equation}
    where $\phi_{\alpha_i}\in C^\infty_0(\mathbb R^d)$ are each a compactly supported smooth function and $\alpha_i\in\mathbb R$. Note that $C$ is defined so that it includes, in addition to functions that decay quickly at infinity, on-shell plane waves as well as their products and derivatives; it is not possible to  restrict to  only on-shell waves since they are not closed under products.

    On $C$, wherever convergent, we set
    \begin{equation}\label{eq:bulkBilinearFormAdS}
        \langle\phi,\phi'\rangle_C\coloneqq\int_0^\infty\mathrm dz\int_{\mathbb R^d}\mathrm d^d\vec y\,z^{-nd-1}\,\phi(z,\vec y)\phi'(z,\vec y),
    \end{equation}
    using the volume density $\sqrt{\det g}=z^{-nd-1}$.

    \paragraph{Boundary function space.}
    For each $\delta\in\mathbb R$, let $C_\partial^\delta$ be (a copy of) the function space
    \begin{equation}
        C_\partial^\delta\coloneqq
        \mathcal C^\infty(\mathbb R^d)
    \end{equation}
    of smooth functions on $\mathbb R^d$.
    This should be thought of as the space of the `values at $z=0$\ of `plane waves'  with `dispersion relation' characterised by $\delta$.
    Let $C_\partial$ be the space of formal sums
    \begin{equation}
        C_\partial = 
        \left\{\sum_{i=1}^\infty(\delta_i,\phi_i)\middle|\phi_i\in C_\partial^{\delta_i}
        \right\}
    \end{equation}
    such that the function
    \begin{equation}
        (z,\vec y)\mapsto\sum_{i=1}^\infty z^{\delta_i}\phi_i(\vec y)
    \end{equation}
    converges pointwise.

    Thus, for the flat space case, $C_\partial$ is the space of linear combinations of on-shell plane waves. The space $C_\partial$ is chosen such that it includes the asymptotic values (i.e.~`values at $z=0$') for solutions to the Helmholtz equation (for both the flat and hyperbolic cases) for different values of the squared mass $m^2$. For the Euclidean case, solutions to the Helmholtz equation are
    \begin{equation}
        z^{\pm\sqrt{m^2+\vec p^2}}\exp(\mathrm i\vec p\cdot\vec y),
    \end{equation}
    which we regard as an element of $C^{\sqrt{m^2+\vec p^2}}_\partial$,
    while for the hyperbolic case, solutions to the Helmholtz equation are of the form
    \begin{subequations}
        \begin{equation}\label{eq:AdS_Helmholtz_ansatz}
            \phi(z,\vec y)=z^{\delta_-}\phi_{\delta_-}(\vec y)+\dotsb+z^{\delta_+}\phi_{\delta_+}(\vec y)+\dotsb,
        \end{equation}
        where
        \begin{equation}\label{eq:conformal_dimension}
            \delta_\pm\coloneqq\tfrac d2\pm\sqrt{\tfrac{d^2}4+m^2}
        \end{equation}
    \end{subequations}
    and $\delta_+$ is the conformal dimension of the corresponding CFT local operator,\footnote{To avoid confusion with the Beltrami Laplacian, we denote the conformal dimension by $\delta_+$.} and $\phi_{\delta_+}$ is related to $\phi_{\delta_-}$ as
    \begin{equation}
        \phi_{\delta_+}(\vec y)\propto\int\mathrm d^d\vec y'\,\frac{\phi_{\delta_-}(\vec y')}{|\vec y-\vec y'|^{2\delta_+}}.
    \end{equation}
    We regard this as the formal sum
    \begin{equation}
        (\delta_-,\phi_{\delta_-})+\cdots+(\delta_+,\phi_{\delta_+})+\cdots\in C_\partial.
    \end{equation}

    Furthermore, on $C_\partial$, wherever convergent, we set
    \begin{equation}\label{eq:boundaryBilinearFormAdS}
        \langle f,g\rangle_{C_\partial}\coloneqq\int_{\mathbb R^d}\mathrm d^d\vec y\,\sum_{\delta\in\mathbb R}f_\delta(\vec y)g_{1+nd-\delta}(\vec y).
    \end{equation}
    We also have the map
    \begin{equation}
        \begin{aligned}
            (-)_\mathrm{pw}\colon C & \rightarrow C_\partial,
            \\
            \left(\phi_\mathrm{interior}+\sum_iz^{\delta_i}\phi_i\right)& \mapsto \sum_i(\delta_i,\phi_i),
        \end{aligned}
    \end{equation}
    where $\phi_\mathrm{interior}$ was defined in \eqref{eq:phi_decomposition}, which picks out the asymptotic components of $\phi\in C$.

    \paragraph{\mathversion{bold}Relative $L_\infty$-algebra.}
    Inspired by our discussion from \Cref{ssec:scalar_field}, the relative $L_\infty$-algebra for a scalar field $\phi$ of mass $m$ with cubic interaction is
    \begin{subequations}\label{eq:relativeLinftyFlatAdS}
        \begin{equation}
            \begin{tikzcd}
                0\rar & C\rar["\mu_1"]\dar["\pi_1"] & C\rar\dar["0"] & 0 
                \\
                0\rar & C_\partial\rar & 0
            \end{tikzcd}
        \end{equation}
        with 
        \begin{equation}
            \begin{gathered}
                \mu_1(\phi)\coloneqq(\Delta-m^2)\phi,
                \quad
                \mu_2(\phi,\phi')\coloneqq-\lambda\phi\phi',
                \\
                \pi_1(\phi)\coloneqq\phi_\text{pw},
            \end{gathered}
        \end{equation}
        where the Beltrami Laplacian is
        \begin{equation}
            \Delta\phi\coloneqq-z^{1+nd}\left(\partial_z(z^{1-nd}\partial_z\phi)+\partial_{\vec y}(z^{-1-nd}z^{2n}\partial_{\vec y}\phi)\right)
        \end{equation}
        and where $c$ is a constant; it will eventually be fixed by requiring that the coefficient for the quadratic term in~\eqref{eq:flat-space-minimal-MC-action} or~\eqref{eq:AdS-minimal-MC-action} below is correctly normalised (when compared to the trivial part of the $S$-matrix or the CFT two-point function). Furthermore,
        \begin{equation}
            \langle\phi,\phi'^+\rangle_\LV\coloneqq\langle\phi,\phi'^+\rangle_C\eqqcolon\langle\phi'^+,\phi\rangle_\LV
            \quad\mbox{and}\quad
            \left\langle\alpha,\alpha'
            \right\rangle_{\LV_\partial}
            \coloneqq\langle\alpha,\beta'\rangle_{C_\partial},
        \end{equation}
    \end{subequations}
    where $\langle-,-\rangle_C$ and $\langle-,-\rangle_{C_\partial}$ were given in~\eqref{eq:bulkBilinearFormAdS} and~\eqref{eq:boundaryBilinearFormAdS}, respectively.

    Using these ingredients, the relative homotopy Maurer--Cartan action~\eqref{eq:relative_MC_action} becomes
    \begin{equation}
        S=\int_0^\infty\frac{\mathrm dz}{z^{1+nd}}\int_{\mathbb R^d}\mathrm d^d\vec y\Big\{\tfrac12\phi(\Delta-m^2)\phi-\tfrac1{3!}\lambda\phi^3\Big\}+c\int_{\mathbb R^d}\mathrm d^d\vec y\,\Big[(z^{-1-nd}\phi(z,\vec y)(\partial_z\phi)(z,\vec y)\Big]_0,
    \end{equation}
    where $[\cdots]_0$ extracts the component that is of order $\mathcal O(z^0)$. This can be recognised as a regularised version of the standard action
    \begin{equation}
        S_\text{naive}=-\int_0^\infty\frac{\mathrm dz}{z^{1+nd}}\int_{\mathbb R^d}\mathrm d^d\vec y\Big\{\tfrac12(\partial\phi)^2+\tfrac12 m^2\phi^2+\tfrac1{3!}\lambda\phi^3\Big\}.
    \end{equation}

    \paragraph{Deformation retract.}
    We have the deformation retract $(i,p,h)$ whose components are 
    \begin{equation}
        \begin{tikzcd}[sep=huge]
            C \ar[rr, shift left, "p" near end]\drar[shift left, "\Delta-m^2"]\ar[dd,"{(-)_\mathrm{pw}}"'] && \mathcal C^\infty(\mathbb R^d)\oplus\mathcal C^\infty(\mathbb R^d)\drar["0"]\ar[dd,"{(f_+,f_-)\mapsto(i(f_+,f_-))_\mathrm{pw}}" near end] \ar[ll, shift left, "i" near start]
            \\
            & C \ular[shift left, "(\Delta-m^2)^{-1}"] \ar[rr, shift left, "p" near start]\ar[dd]&&\mathcal C^\infty(\mathbb R^d)\oplus\mathcal C^\infty(\mathbb R^d)\ar[dd] \ar[ll, shift left, "i" near end]
            \\
            C_\partial \ar[rr, shift left, "\mathrm{id}"]\drar[shift left] && C_\partial\drar\ar[ll, shift left, "\mathrm{id}"]
            \\
            & 0 \ar[rr, shift left] \ular[shift left] &&0. \ar[ll, shift left]
            \\
        \end{tikzcd}
    \end{equation}
    Here, $i$ is given (for both fields and antifields) as
    \begin{equation}
        i\,:\,(f_+,f_-)\mapsto\int\mathrm d^d\vec y\,(f_+(\vec y)K_+(-,-;\vec y)+f_-(\vec y)K_-(-,-;\vec y))
    \end{equation}
    in terms of the bulk-to-boundary propagator
    \begin{equation}
        K_\pm(\vec y,z;\vec y')\coloneqq
        \begin{cases}
            \int\mathrm d^d\vec p\,z^{\delta_\pm(\vec p)}\exp(\mathrm i(\vec y-\vec y')\cdot\vec p)&n=0\\
            \frac{\Gamma(\delta_\pm)}{\pi^{d/2}\Gamma(\delta_\mp)}\left(\frac z{z^2+(\vec y-\vec y')^2}\right)^{\delta_\pm}&n=1,
        \end{cases}
    \end{equation}
    where
    \begin{equation}
        \delta_\pm(\vec p)\coloneqq
        \begin{cases}
            \pm\sqrt{m^2+\vec p^2} & \text{for}\quad n=0
            \\
            \frac d2\pm\sqrt{\frac{d^2}4+m^2} & \text{for}\quad n=1
        \end{cases}
    \end{equation}
    is either the on-shell energy (and its negative) for flat space or the conformal dimension (and its conjugate) of the corresponding CFT operator for hyperbolic space, and where
    \begin{equation}
        m^2\geq
        \begin{cases}
            0 & \text{for}\quad n=0
            \\
            -\frac{d^2}4 & \text{for}\quad n=1
        \end{cases}
    \end{equation}
    is the squared mass of the particle, which obeys the Breitenlohner--Freedman bound \cite{Breitenlohner:1982bm,Breitenlohner:1982jf} in the case of $n=1$ and is non-negative for $n=0$.     
    For the flat space case, the expression is perhaps clearer in momentum space:
    \begin{equation}
        \hat K_\pm(\vec y,z;\vec p)
        =z^{\delta_\pm(\vec p)}\exp(\mathrm i\vec y\cdot\vec p)
        =\exp\left(\delta_\pm(\vec p)t+\mathrm i\vec y\cdot\vec p\right),
    \end{equation}
    which is the Wick-rotated version of an on-shell plane wave.
    
    The projection map $p$ is given in terms of a suitable left inverse of $i$.

    \paragraph{Minimal model for flat space.}
    Let us Fourier-transform $\vec y$ into $\vec p$. Then, solutions to the equations of motion are linear combinations of plane waves of the form
    \begin{subequations}\label{planebasis}
        \begin{equation}
            \hat\phi_\pm(z,\vec p)=z^{\pm E_{\vec p}}\hat\phi^\pm_{E_{\vec p}}(\vec p)
        \end{equation}
        where
        \begin{equation}
            E_{\vec p}\coloneqq\sqrt{m^2+\vec p^2}
        \end{equation}
    \end{subequations}
    is the mass-shell condition.
    
    The cochain complex underlying the minimal model for~\eqref{eq:relativeLinftyFlatAdS} is then
    \begin{equation}
        \begin{tikzcd}
            0\rar&\mathcal C^\infty(\mathbb R^d)\oplus\mathcal C^\infty(\mathbb R^d)\rar["{\overset\circ\mu}_1=0"]\dar["{(f_+,f_-)\mapsto (i(f_+,f_-))_\mathrm{pw}}"']&\mathcal C^\infty(\mathbb R^d)\oplus\mathcal C^\infty(\mathbb R^d)\rar\dar["0"] & 0
            \\
            0\rar & C_\partial\oplus C_\partial\rar & 0
        \end{tikzcd}
    \end{equation}
    where all the $\overset\circ\mu_{n>1}$ are present and are given by the homological perturbation lemma as discussed in \Cref{ssec:minimal_model}. Furthermore, the relative homotopy Maurer--Cartan action~\eqref{eq:relative_MC_action} for the minimal model is
    \begin{equation}\label{eq:flat-space-minimal-MC-action}
        S=\frac12\int_{\mathbb R^d}\mathrm d^d\vec p\,E_{\vec p}\,\hat\phi^+_{E_{\vec p}}(\vec p)\hat\phi^-_{E_{\vec p}}(-\vec p)-\int_0^\infty\frac{\mathrm dz}{z}\int_{\mathbb R^d}\mathrm d^d\vec y\frac1{3!}\lambda\phi^3(z,\vec y)+\cdots,
    \end{equation}
    where $\phi(z,\vec y)$ is a linear combination of the Fourier-transform of~\eqref{planebasis} and the ellipsis denotes higher-order scattering amplitudes via the higher-order products $\overset\circ\mu_{n>1}$. In particular, the two-point scattering amplitude is seen to reproduce the (Euclidean) Klein--Gordon metric (see~e.g.~\cite[(4.9)]{Jacobson:2003vx}), correctly pairing positive-energy (incoming) and negative-energy (outgoing) states.

    \begin{table}[h]
        \vspace{15pt}
        \begin{center}
            \begin{tabular}{cc}
                \toprule
                Witten Diagrams & Homotopy Transfer
                \\
                \midrule
                bulk--bulk propagator & homotopy $h$
                \\
                bulk--boundary propagator & $i$ and $p$
                \\
                boundary--boundary propagator & $\langle\pi_1(-),\pi_1(-)\rangle$
                \\
                \bottomrule
            \end{tabular}
        \end{center}
        \caption{Correspondence between Witten diagrams and homotopy algebras.}\label{table:WittenDiagramHomotopyTransfer}
    \end{table}
    
    \paragraph{Minimal model for anti-de~Sitter space.}
    Let us take the minimal model for anti-de~Sitter space, that is, for $n=1$. In this case, solutions to the Helmholtz equation follow the ansatz \eqref{eq:AdS_Helmholtz_ansatz}.

    As before, the cochain complex underlying the minimal model for~\eqref{eq:relativeLinftyFlatAdS} is
    \begin{equation}
        \begin{tikzcd}
            0\rar&\mathcal C^\infty(\mathbb R^d)\oplus\mathcal C^\infty(\mathbb R^d)\rar["{\overset\circ\mu}_1=0"]\dar["{(f_+,f_-)\mapsto (i(f_+,f_-))_\mathrm{pw}}"']&\mathcal C^\infty(\mathbb R^d)\oplus\mathcal C^\infty(\mathbb R^d)\rar\dar["0"] & 0
            \\
            0\rar & C_\partial\oplus C_\partial\rar & 0
        \end{tikzcd}
    \end{equation}
    with $\overset\circ\mu_{n>1}$ given by the homological perturbation lemma as discussed in \Cref{ssec:minimal_model}. Furthermore, the relative homotopy Maurer--Cartan action~\eqref{eq:relative_MC_action} for the minimal model is
    \begin{equation}\label{eq:AdS-minimal-MC-action}
        S=\frac12\int_{\mathbb R^{2d}}\mathrm d\vec y\,\mathrm d\vec y'\,\frac{\phi_{\delta_+}(\vec y)\phi_{\delta_+}(\vec y')}{|\vec y-\vec y'|^{2\delta_+}}-\int_0^\infty\frac{\mathrm dz}{z^{1+d}}\int_{\mathbb R^d}\mathrm d^d\vec y\frac1{3!}\lambda\phi^3(z,\vec y)+\cdots
    \end{equation}
    where the ellipsis encodes the higher-order $\overset\circ\mu_{n>1}$ generated by homotopy transfer, corresponding to evaluating the sum over connected Witten diagrams; these encode the connected $(n+1)$-point correlation functions of the boundary CFT. Then, it is clear that the corresponding scattering amplitudes reproduce those of Witten diagrams:
    \begin{subequations}
        \begin{equation}
            \begin{tikzpicture}[
                scale=0.4,
                baseline={([yshift=-0.5ex]current bounding box.center)}
            ]
                \draw[color=black!60, fill=white!5, very thick](0,0) circle (2);
                \draw (0,0) -- (canvas polar cs:radius=2cm,angle=30);
                \draw (0,0) -- (canvas polar cs:radius=2cm,angle=150);
                \draw (0,0) -- (canvas polar cs:radius=2cm,angle=270);
                \draw[fill=black] (0,0) circle (0.1);
            \end{tikzpicture}
            =
            \begin{tikzpicture}[
                scale=1.1,
                baseline={([yshift=-0.5ex]current bounding box.center)}
            ]
                \begin{feynhand}
                    \vertex (a) at (-1,-1) {$i$};
                    \vertex (b) at (1,-1) {$i$};
                    \vertex (c) at (0,1) {$p$};
                    \vertex [ringdot] (o) at (0,0) {$\mu_2$};
                    \propag (a) to (o);
                    \propag (b) to (o);
                    \propag (c) to (o);
                \end{feynhand}
            \end{tikzpicture}
        \end{equation}
        \begin{multline}
            \begin{tikzpicture}[
                scale=0.4,
                baseline={([yshift=-0.5ex]current bounding box.center)}
            ]
                \draw[color=black!60, fill=white!5, very thick](0,0) circle (2);
                \draw (0,0) -- (canvas polar cs:radius=2cm,angle=45);
                \draw (0,0) -- (canvas polar cs:radius=2cm,angle=135);
                \draw (0,0) -- (canvas polar cs:radius=2cm,angle=-135);
                \draw (0,0) -- (canvas polar cs:radius=2cm,angle=-45);
                \draw[fill=black] (0,0) circle (0.1);
            \end{tikzpicture}
            +
            \begin{tikzpicture}[
                scale=0.4,
                baseline={([yshift=-0.5ex]current bounding box.center)}
            ]
                \draw[color=black!60, fill=white!5, very thick](0,0) circle (2);
                \draw (0.5,0) -- (canvas polar cs:radius=2cm,angle=45);
                \draw (-0.5,0) -- (canvas polar cs:radius=2cm,angle=135);
                \draw (-0.5,0) -- (canvas polar cs:radius=2cm,angle=-135);
                \draw (0.5,0) -- (canvas polar cs:radius=2cm,angle=-45);
                \draw (-0.5,0) -- (0.5,0);
                \draw[fill=black] (0.5,0) circle (0.1);
                \draw[fill=black] (-0.5,0) circle (0.1);
            \end{tikzpicture}
            +
            \begin{tikzpicture}[
                scale=0.4,
                baseline={([yshift=-0.5ex]current bounding box.center)}
            ]
                \draw[color=black!60, fill=white!5, very thick](0,0) circle (2);
                \draw (0,0.5) -- (canvas polar cs:radius=2cm,angle=45);
                \draw (0,0.5) -- (canvas polar cs:radius=2cm,angle=135);
                \draw (0,-0.5) -- (canvas polar cs:radius=2cm,angle=-135);
                \draw (0,-0.5) -- (canvas polar cs:radius=2cm,angle=-45);
                \draw (0,0.5) -- (0,-0.5);
                \draw[fill=black] (0,0.5) circle (0.1);
                \draw[fill=black] (0,-0.5) circle (0.1);
            \end{tikzpicture}
            +
            \begin{tikzpicture}[
                scale=0.4,
                baseline={([yshift=-0.5ex]current bounding box.center)}
            ]
                \draw[color=black!60, fill=white!5, very thick](0,0) circle (2);
                \draw (0,0.5) -- (canvas polar cs:radius=2cm,angle=45);
                \draw (0,-0.5) -- (canvas polar cs:radius=2cm,angle=135);
                \draw (0,0.5) -- (canvas polar cs:radius=2cm,angle=-135);
                \draw (0,-0.5) -- (canvas polar cs:radius=2cm,angle=-45);
                \draw (0,0.5) -- (0,-0.5);
                \draw[fill=black] (0,0.5) circle (0.1);
                \draw[fill=black] (0,-0.5) circle (0.1);
            \end{tikzpicture}
            \\=
            \begin{tikzpicture}[
                scale=1.15,
                baseline={([yshift=-0.5ex]current bounding box.center)}
            ]
                \begin{feynhand}
                    \vertex (a) at (-1,-1) {$i$};
                    \vertex (b) at (1,-1) {$i$};
                    \vertex (c) at (0,-1) {$i$};
                    \vertex (out) at (0,1) {$p$};
                    \vertex [ringdot] (centre) at (0,0) {$\mu_3$};
                    \propag (a) to (centre);
                    \propag (b) to (centre);
                    \propag (c) to (centre);
                    \propag (centre) to (out);
                \end{feynhand}
            \end{tikzpicture}
            +
            \begin{tikzpicture}[
                scale=1.15,
                baseline={([yshift=-0.5ex]current bounding box.center)}
            ]
                \begin{feynhand}
                    \vertex (a) at (-1,-1) {$i$};
                    \vertex (b) at (0,-1) {$i$};
                    \vertex (c) at (1,-1) {$i$};
                    \vertex (out) at (0,1) {$p$};
                    \vertex [ringdot] (centre) at (0,0.3) {$\mu_2$};
                    \vertex [ringdot] (centre2) at (-0.3,-0.3) {$\mu_2$};
                    \propag (a) to (centre2);
                    \propag (b) to (centre2);
                    \propag (centre) to [edge label'={$h$}] (centre2);
                    \propag (c) to (centre);
                    \propag (centre) to (out);
                \end{feynhand}
            \end{tikzpicture}
            \\
            +
            \begin{tikzpicture}[
                scale=1.15,
                baseline={([yshift=-0.5ex]current bounding box.center)}
            ]
                \begin{feynhand}
                    \vertex (a) at (-1,-1) {$i$};
                    \vertex (b) at (0,-1) {$i$};
                    \vertex (c) at (1,-1) {$i$};
                    \vertex (out) at (0,1) {$p$};
                    \vertex [ringdot] (centre) at (0,0.3) {$\mu_2$};
                    \vertex [ringdot] (centre2) at (0.3,-0.3) {$\mu_2$};
                    \propag (a) to (centre);
                    \propag (b) to (centre2);
                    \propag (centre) to [edge label={$h$}] (centre2);
                    \propag (c) to (centre2);
                    \propag (centre) to (out);
                \end{feynhand}
            \end{tikzpicture}
            +
            \begin{tikzpicture}[
                scale=1.15,
                baseline={([yshift=-0.5ex]current bounding box.center)}
            ]
                \begin{feynhand}
                    \vertex (a) at (-1,-1) {$i$};
                    \vertex (b) at (0,-1) {$i$};
                    \vertex (c) at (1,-1) {$i$};
                    \vertex (out) at (0,1) {$p$};
                    \vertex [ringdot] (centre) at (0.3,0.3) {$\mu_2$};
                    \vertex [ringdot] (centre2) at (-0.3,-0.3) {$\mu_2$};
                    \propag (a) to (centre2);
                    \propag (b) to (centre);
                    \propag (centre) to [edge label'={$h$}](centre2);
                    \propag (c) to (centre2);
                    \propag (centre) to (out);
                \end{feynhand}
            \end{tikzpicture}
        \end{multline}
    \end{subequations}
    and so on, where in the homological perturbation lemma we have the correspondence in \Cref{table:WittenDiagramHomotopyTransfer}.
    
    The CFT two-point correlator is \emph{not} given by the usual homotopy transfer but is instead given by $\langle\pi_1(-),\pi_1(-)\rangle$, that is, from the boundary term of the relative Maurer--Cartan action; this is essentially the classic derivation~\cite[§2.4,§2.5]{Witten:1998qj} that involves integration by parts of the bulk AdS action to reduce it to a boundary term.\footnote{In~\cite{Witten:1998qj} this computation is done with a bulk-to-boundary propagator, which provides the required regularisation; the choice of the relative $L_\infty$-algebra here encodes an equivalent choice of regulator.} This seeming inhomogeneity is, in fact, natural: the two-point correlator is \emph{not} part of the connected correlator, since the full correlator is obtained by `exponentiating' the connected correlator, whose `identity' component consists purely of Wick contractions involving the two-point correlator (and other Wick contractions).
    
    In this sense, the leftmost diagram in~\cite[Figure~1]{Freedman:1998tz}, 
    \begin{equation}
     \begin{tikzpicture}[
                scale=0.5,
                baseline={([yshift=-0.5ex]current bounding box.center)}
            ]
                \draw[color=black!60, fill=white!5, very thick](0,0) circle (2);
                \draw (0,0) -- (canvas polar cs:radius=2cm,angle=0);
                \draw (0,0) -- (canvas polar cs:radius=2cm,angle=180);
                \draw[fill=black] (0,0) circle (0.1);
            \end{tikzpicture},
        \end{equation}
         which is often used to represent the two-point function,
    \begin{equation}
            \langle O(\vec y_1), O(\vec y_2)\rangle_\mathrm{CFT}
            =
            \int_\mathrm{AdS}\mathrm d\vec y\frac{\mathrm dz}{z^{1+d}}\Big\{\partial K(\vec y,z;\vec y_1)\cdot \partial K(\vec y,z;\vec y_2)+m^2K(\vec y,z;\vec y_1) K(\vec y,z,\vec y_2)\Big\},
    \end{equation}
    is misleading in that the `vertex', which looks like it should be $\mu_1$, in fact is not (since we have picked up a boundary term); it can, instead, be interpreted as the AdS boundary-to-boundary propagator.

    \subsection{Gauge theory on a manifold-with-boundary}

    Here, we consider the relative $L_\infty$-algebras for Chern-Simons and Yang--Mills theory on  oriented compact Riemannian manifolds  with boundary. See also~\cite{chiaffrino2023holography} for an $L_\infty$-algebra approach to Yang--Mills theory on a manifold with boundary .  

    \subsubsection{Chern--Simons theory}

    \paragraph{BV action.}
    Consider an ordinary finite-dimensional metric Lie algebra $(\LV,[-,-]_\LV,\langle-,-\rangle_\LV)$. The na\"ive\footnote{In the sense that it is derived directly from the canonical symplectic $Q$-manifold $(\Omega^\bullet(M,\LV)[1],Q,\omega)$. See~\cite{Jurco:2018sby} for a detailed discussion of its $L_\infty$-algebra realisation.} homotopy Maurer--Cartan Chern--Simons BV action on an oriented compact Riemannian manifold $(M,g)$ with boundary $\partial M$ is given by 
    \begin{equation}\label{eq:CSaction}
        \begin{aligned}
            S^\mathrm{CS}_\mathrm{hMC}&\coloneqq\int_M\Big\{\tfrac12\langle A,\mathrm d_MA\rangle_\LV+\tfrac{1}{3!}\langle A,[A,A]_\LV\rangle_\LV
            \\
            &\kern1.5cm-\tfrac12\langle A^+,\mathrm d_Mc\rangle_\LV-\tfrac12\langle c,\mathrm d_MA^+\rangle_\LV-\langle A^+,[A,c]_\LV\rangle_\LV+\tfrac12\langle c^+,[c,c]_\LV\rangle_\LV\Big\},
        \end{aligned}
    \end{equation}
    where $c\in\Omega^0(M,V)$ is the ghost and $A\in\Omega^1(M,V)$ the Chern--Simons gauge potential and $A^+\in\Omega^2(M,V)$ and $c^+\in\Omega^3(M,V)$ the corresponding anti-fields. This action is invariant under the corresponding BV transformations,    
    \begin{equation}
        \begin{aligned}
            Q_\mathrm{BV}c&\coloneqq-\tfrac12[c,c]_\LV,
            \\
            Q_\mathrm{BV}A&\coloneqq\nabla_Mc,
            \\
            Q_\mathrm{BV}A^+&\coloneqq-F-[c,A^+]_\LV,
            \\
            Q_\mathrm{BV}c^+&\coloneqq\nabla_MA^+-[c,c^+]_\LV,
        \end{aligned}
    \end{equation}
    where $\nabla_M\coloneqq\mathrm d_M+[A,-]_\LV$ and $F\coloneqq\mathrm d_MA+\frac12[A,A]_\LV$, up to a boundary term which is given by 
    \begin{equation}
        Q_\mathrm{BV}S^\mathrm{CS}_\mathrm{hMC}=\int_{\partial M}\Big\{\langle c,\mathrm d_MA\rangle_\LV+\tfrac14\langle[c,c]_\LV,A^+\rangle_\LV+\tfrac14\langle c,[A,A]_\LV\rangle_\LV\Big\}.
    \end{equation}

    Note that one can integrate by parts to write~\eqref{eq:CSaction} as 
    \begin{subequations}
        \begin{equation}
            S^\mathrm{CS}_\mathrm{hMC}=S^\mathrm{CS}_\mathrm{BV}-\tfrac12\int_{\partial M}\langle c,A^+\rangle_\LV,
        \end{equation}
        where 
        \begin{equation}\label{eq:covCS}
            S^\mathrm{CS}_\mathrm{BV}=\int_M\Big\{\tfrac12\langle A,\mathrm d_MA\rangle_\LV+\tfrac{1}{3!}\langle A,[A,A]_\LV\rangle_\LV-\langle A^+,\nabla_Mc\rangle_\LV+\tfrac12\langle c^{+},[c,c]_\LV\rangle_\LV\Big\}.
        \end{equation}
    \end{subequations}
    Note also that the inner product used to define the action is only cyclic up to a boundary term, 
    \begin{equation}
        \int_M\langle A_1,\mathrm d_MA_2\rangle_\LV=\int_M\langle A_2,\mathrm d_MA_1\rangle_\LV-\int_{\partial M}\langle A_1,A_2\rangle_\LV,
    \end{equation}
    for all $A_{1,2}\in\Omega^1(M,V)$. This must be rectified by a boundary term deriving from the relative $L_\infty$-algebra. However, this additional term drops out of the relative homotopy Maurer--Cartan action, since $\int_{\partial M}\langle A,A\rangle_\LV=0$, so it is only visible at the level of the relative $L_\infty$-algebra.
    
    \paragraph{\mathversion{bold}Relative $L_\infty$-algebra.}
    The cochain complex of the relative $L_\infty$-algebra is given by $\LV$-valued $p$-forms on $M$ and $\partial M$,
    \begin{subequations}
        \begin{equation}
            \begin{tikzcd}
                \overbrace{\Omega^0(M,\LV)}^{\in\,c}\rar["\mu_1"]\dar["\pi_1"] & \overbrace{\Omega^1(M,\LV)}^{\in\,A}\rar["\mu_1"]\dar["\pi_1"] & \overbrace{\Omega^2(M,\LV)}^{\in\,A^+}\rar["\mu_1"]\dar["\pi_1"] & \overbrace{\Omega^3(M,\LV)}^{\in\,c^+} 
                \\
              \underbrace{\Omega^0(\partial M,\LV)}_{\in\,\gamma}  \rar["\mu_1^\partial"] & \underbrace{\Omega^1(\partial M,\LV)}_{\in\,\alpha}\rar["\mu_1^\partial"] & \underbrace{\Omega^2(\partial M,\LV)}_{\in\,\alpha^+} 
            \end{tikzcd}
        \end{equation}
        with 
        \begin{equation}
            \begin{gathered}
                \mu_1(c)\coloneqq\mathrm d_Mc,
                \quad
                \mu_1(A)\coloneqq\mathrm d_MA,
                \quad
                \mu_1(A^+)\coloneqq\mathrm d_MA^+,
                \\
                \mu_2(c,c')\coloneqq[c,c']_\LV,
                \quad
                \mu_2(c,A)\coloneqq[c,A]_\LV,
                \quad
                \mu_2(c,A^+)\coloneqq[c,A^+]_\LV,
                \quad
                \mu_2(c,c'^+)\coloneqq[c,c'^+]_\LV,
                \\
                \mu_2(A,A')\coloneqq [A,A']_\LV,
                \quad
                \mu_2(A,A'^+)\coloneqq[A,A'^+]_\LV.
            \end{gathered}
        \end{equation}
        and  
        \begin{equation}
            \pi_1(c)\coloneqq c|_{\partial M}, 
            \quad
            \pi_1(A)\coloneqq A|_{\partial M},
            \quad
            \pi_1(A^+)\coloneqq A^+|_{\partial M}
        \end{equation}
        and
        \begin{equation}
            \begin{gathered}
                \mu^\partial_1(\gamma)\coloneqq\mathrm d_{\partial M}\gamma,
                \quad
                \mu^\partial_1(\alpha)\coloneqq\mathrm d_{\partial M}\alpha,
                \quad
                \mu^\partial_1(\alpha^+)\coloneqq\mathrm d_{\partial M}\alpha^+,
                \\
                \mu^\partial_2(\gamma,\gamma')\coloneqq[\gamma,\gamma']_\LV,
                \quad
                \mu^\partial_2(\gamma,\alpha)\coloneqq[\gamma,\alpha]_\LV,
                \quad
                \mu^\partial_2(\gamma,\alpha^+)\coloneqq[\gamma,\alpha^+]_\LV,
                \\
                \mu^\partial_2(\alpha,\alpha')\coloneqq [\alpha,\alpha']_\LV.
            \end{gathered}
        \end{equation}        
        In addition, we introduce the degree $-2$ (before grade-shiting) bilinear form has non-vanishing components
        \begin{equation}
            \begin{gathered}
                \langle\gamma,\alpha^+\rangle_{\partial}\coloneqq\int_{\partial M}\langle\gamma,\alpha^+\rangle_\LV
                \quad\mbox{and}\quad
                \left\langle
                    \alpha, 
                    \alpha'
                \right\rangle_{\partial}
                \coloneqq\int_{\partial M}\langle\alpha,\alpha'\rangle_\LV.
            \end{gathered}
        \end{equation}
    \end{subequations}
    
    With these definitions, the relative homotopy Maurer--Cartan action~\eqref{eq:relative_MC_action} becomes
    \begin{equation}
        \begin{aligned}
            S^\mathrm{CS}_\mathrm{rhMC}&=\int_M\Big\{\tfrac12\langle A,\mathrm d_MA\rangle_\LV+\tfrac{1}{3!}\langle A,[A,A]_\LV\rangle_\LV
            \\
            &\kern1cm-\tfrac12\langle A^+,\mathrm d_Mc\rangle_\LV-\tfrac12\langle c,\mathrm d_MA^+\rangle_\LV-\langle A^+,[A,c]_\LV\rangle_\LV+\tfrac12\langle c^+,[c,c]_\LV\rangle_\LV\Big\}\Big|_{\partial M}
            \\
            &\kern1cm-\tfrac12\int_{\partial M}\Big\{\langle c,A^+\rangle_\LV-\langle A,A\rangle_\LV\Big\}
            \\
            &=\int_M\Big\{\tfrac12\langle A,\mathrm d_MA\rangle_\LV+\tfrac{1}{3!}\langle A,[A,A]_\LV\rangle_\LV\langle A^+,\nabla_Mc\rangle_\LV+\tfrac12\langle c^+,[c,c]_\LV\rangle_\LV\Big\}.
        \end{aligned}
    \end{equation}
   
    \paragraph{\mathversion{bold}Relative minimal model.}
    To construct the relative minimal model, we first consider the cohomology of the bulk cochain complex, $(\Omega^\bullet(M,\LV),\mathrm d_M)$. By the K\"unneth formula 
    \begin{equation}
        H_{\mu_1}^\bullet(\Omega^{\bullet}(M,\LV)\cong\bigoplus_{p=0}^{3}H_\mathrm{dR}^p(M;\mathbb{R})\otimes\LV,
    \end{equation}
    where $H_\mathrm{dR}^\bullet(M;\mathbb{R})$ is the de Rham cohomology and for $M$ an oriented compact Riemannian manifold with boundary $\partial M$ (and possibly corners) we have $H_\mathrm{dR}^\bullet(M;\mathbb{R})\cong H^\bullet(M;\mathbb{R})$, the real singular cohomology of $M$~\cite{SAMELSON1967427}. Thus, the minimal model absent interactions (higher $L_\infty$-brackets) computes the real singular cohomology of $M$. 

    To identify $H_\mathrm{dR}^\bullet(M;\mathbb{R})$ explicitly, denote the closed, exact, co-closed, co-exact forms and their intersections by (leaving $M$ implicit)
    \begin{equation}
        \begin{aligned}
            C^p&\coloneqq\{\omega\in\Omega^p(M)\mid\mathrm d_M\omega=0\},
            \\
            cC^p&\coloneqq\{\omega\in\Omega^p(M)\mid\mathrm d_M^\dagger\omega=0\},
            \\
            E^p&\coloneqq\{\omega\in\Omega^p(M)\mid\omega=\mathrm d_M\eta\},
            \\
            cE^p&\coloneqq\{\omega\in\Omega^p(M)\mid\omega=\mathrm d_M^\dagger\eta\},
            \\
            CcC^p&\coloneqq C^p\cap cC^p,
            \\
            EcC^p&\coloneqq E^p\cap cC^p\subseteq CcC^p,
            \\
            CcE^p&\coloneqq C^p\cap cE^p \subseteq CcC^p.
        \end{aligned}
    \end{equation}

    We shall also need to impose Dirichlet (relative) $D$ and Neumann (absolute) $N$ boundary conditions on the space of $p$-forms. To do so, it is convenient to introduce local coordinates $x=(y,r)$ on $M$ near $\partial M$, where $y$ are local coordinates on $\partial M$ and  $r\geq 0$ is the normal distance to the boundary so that for $p\in \partial M$ we have $r(p)=0$. We denote forms at a boundary point $p$ by $\omega|_{p\in \partial M}$\footnote{Not to be confused with the restriction given by the pull-back of the inclusion, denoted by $\omega|_{\partial M}\coloneqq\iota^* \omega$.}, which decomposes into tangential and normal components,
    \begin{equation}
        \omega|_{p\in \partial M}=\omega_p^{\|}+\omega_p^\perp.
    \end{equation}
    In  local coordinates, this decomposition can be written
    \begin{equation}
        \omega|_{p\in \partial M}=\alpha_p^{\|}+\alpha_p^\perp\wedge\mathrm{d}r
        \quad\mbox{and}\quad
        \alpha_p^{\|},\alpha_p^\perp\in\Omega^\bullet(\partial M).
    \end{equation}
    In a coordinate-free language, $\omega_p^{\|}$ is  given by $\omega_p^{\|}(X_1,\ldots X_p)\coloneqq\omega(X_1^{\|},\ldots X_p^{\|})$ for all $X_1,\ldots X_p\in \Gamma(M|_{\partial M},TM)$ and $X=X^{\perp}+X^{\|}$ denotes the decomposition into tangential and normal parts. The normal component is then defined by $\omega_p^\perp\coloneqq\omega_p^{\|}-\omega|_{p\in \partial M}$. For notational clarity we will henceforth write $\omega^{\|}$ and $\omega^\perp$ for the tangential and normal components at any boundary point.

    The tangential and norm components define the $D$ (relative) and $N$ (absolute) boundary conditions,
    \begin{equation}
        \Omega_D^p(M)\coloneqq\{\omega\in\Omega^p(M)\mid\omega^{\|}=0\}
        \quad\mbox{and}\quad
        \Omega_N^p(M)\coloneqq\{\omega\in\Omega^p(M)\mid\omega^\perp=0\},
    \end{equation}
    where for (co-)exact forms the boundary conditions are applied to the pre-images
    \begin{equation}
        cE_N^p\coloneqq\mathrm d_M^\dagger\Omega_N^{p+1}(M)
        \quad\mbox{and}\quad
        E_D^p\coloneqq\mathrm d_M\Omega_D^{p-1}(M).
    \end{equation}

    With these conventions, the Hodge decomposition for an oriented compact, connected, smooth Riemannian manifold with boundary is~\cite{Morrey:1956}
    \begin{equation}\label{HD}
        \Omega^p(M)\cong cE_N^p\oplus CcC^p\oplus E_D^p\cong c E_N^p\oplus CcC_N^p\oplus EcC^p\oplus E_D^p.
    \end{equation}
    Since $EcC^p\oplus E_D^p\cong E^p$, from the above, $CcC_N^p$ is the orthogonal complement of the exact forms inside the closed forms,
    \begin{equation}
        C^p\cong CcC_N^p\oplus E^p,
    \end{equation}
    and we have $H_\mathrm{dR}^p(M;\mathbb{R})\cong CcC_N^p\cong H^p(M;\mathbb{R})$. See for example~\cite{gilkey1984invariance}. Note also that the relative $D$ boundary condition applied to $CcC^p$ gives the relative cohomology, $CcC_D^p\cong H^p(M,\partial M;\mathbb{R})$, cf.~\cite{gilkey1984invariance}.

    Recall that when the boundary is empty, $CcC^p\cong\ker(\Delta|_{\Omega^p(M)})$ by Poincar\'e duality and so, $\Omega^p(M)\cong CcC^p\oplus\Delta\Omega^p(M)$. Thus, the spaces of harmonic \emph{fields}\footnote{To use the nomenclature introduce by Kodaira in~\cite{Kodaira:1949}.}, $CcC^p$, and harmonic \emph{forms}, $\ker(\Delta|_{\Omega^p(M)})$, are isomorphic. However, when the boundary is non-empty there may be more harmonic \emph{forms} than harmonic \emph{fields} and $\Delta\Omega^p(M)\cong\Omega^p(M)$. See, e.g.~\cite{Schwarz1995}. Indeed, while $CcC_N^p$ and $CcC_D^p$ are finite dimensional, $CcC^p$ and $\ker(\Delta|_{\Omega^p(M)})$ are infinite dimensional for $0<p$.

    The deformation retract data 
    \begin{equation}
        \begin{tikzcd}
            (\Omega^{\bullet}(M, \LV),\mathrm d_M)\ar[loop left, "h"]\rar[shift left, "p_1"] & (CcC_N^\bullet,0) \lar[shift left, "i_1"]
        \end{tikzcd}
    \end{equation}
    is given by the trivial embedding $i_1$ and canonical projection $p_1$ given by the Hodge decomposition.  

    We define the contracting homotopy by
    \begin{equation}
        h\coloneqq\mathrm d_M^{-1}P_E,
    \end{equation}
    where $P_E$ is the projector onto the image of $\mathrm d_M$ and $\mathrm d_M^{-1}$ is the inverse of the operator $\mathrm d_M$ restricted to the  orthogonal complement of its kernel, cf.~the definition of the Chern--Simons Green function in~\cite{Axelrod:1991vq}. The Hodge decomposition~\eqref{HD} then implies
    \begin{equation}
        \mathrm{id}-h\circ\mathrm d_M-\mathrm d_M\circ h=P_{CcC_N^\bullet},
    \end{equation}
    where $P_{CcC_N^\bullet}$ is the projector onto $CcC_N^\bullet\subseteq\Omega^\bullet(M)$. Explicitly, using~\eqref{HD} any $p$-form can be written 
    \begin{equation}
        \omega=\mathrm d_M\alpha+\mathrm d_M^\dagger\beta+\mathrm d_M\gamma+\theta,
    \end{equation}
    where $\mathrm d_M\alpha\in E^p_D$, $\mathrm d_M^\dagger\beta\in cE^p_N$, $\gamma\in cC^{p-1}$, and $\theta\in CcC^p_N$. Consequently,
    \begin{equation}
        (\mathrm d_M\circ h)(\omega)=\mathrm d_M\alpha+\mathrm d_M\gamma
        \quad\mbox{and}\quad
        (h\circ\mathrm d_M)(\omega)=\mathrm d_M^\dagger\beta
    \end{equation}
    and
    \begin{equation}
        (\mathrm{id}-h\circ\mathrm d_M-\mathrm d_M\circ h)(\omega)=\theta,
    \end{equation}
    as required. 

    The relative minimal model is  given by 
    \begin{equation}
        \begin{tikzcd}
            {CcC^0(M,\LV)}\rar["0"]\dar["\overset\circ\pi_1"] & {CcC_N^1(M,\LV)}\rar["0"]\dar["\overset\circ\pi_1"] & {CcC_N^2(M,\LV)}\rar["0"]\dar["\overset\circ\pi_1"] & {CcC_N^3(M,\LV)}
            \\
            {C^0(\partial M,\LV)}\rar["0"] & {CcC^1(\partial M,\LV)}\rar["0"] & {CcC^2(\partial M,\LV)}
        \end{tikzcd},
    \end{equation}
    where $\overset\circ\pi_1(-)=-|_{\partial M}$ which is uniquely determined by the tangential component $\omega^{\|}$ (although they strictly speaking belong to different spaces) and so for all $\omega\in CcC^p_N$ we may formally identify $\omega|_{p\in\partial M}=\omega^{\|}=\omega|_{\partial M}$. The higher products follow from the homological perturbation lemma as discussed in \Cref{ssec:minimal_model} and give a perturbative expansion of the classical solutions given specified boundary data. 

    \subsubsection{Yang--Mills theory}

    Finally, let us discuss Yang--Mills theory in four dimensions, including a topological $\theta$-term, in the framework of relative $L_\infty$-algebras. We work with the usual second-order formulation.\footnote{For first-order formulations see e.g.~\cite{Cattaneo_2014,Jurco:2018sby,Macrelli:2019afx}.}

    \paragraph{\mathversion{bold}Relative $L_\infty$-algebra.}
    We take the colour Lie algebra of the theory to be an ordinary finite-dimensional metric Lie algebra $(\LV,[-,-]_\LV,\langle-,-\rangle_\LV)$. Correspondingly, the relative $L_\infty$-algebra of Yang--Mills theory on an oriented compact four-dimensional Riemannian manifold $(M,g)$ with boundary $\partial M$ is given by
    \begin{subequations}
        \begin{equation}
            \begin{tikzcd}
                \overbrace{\Omega^0(M,\LV)}^{\in\,c}\rar["\mu_1"]\dar["\pi_1"] & \overbrace{\Omega^1(M,\LV)}^{\in\,A}\rar["\mu_1"]\dar["\pi_1"] & \overbrace{\Omega^3(M,\LV)}^{\in\,A^+}\rar["\mu_1"]\dar["\pi_1"] & \overbrace{\Omega^4(M,\LV)}^{\in\,c^+}
                \\
                \underbrace{\Omega^0(\partial M,\LV)}_{\in\,\gamma}\rar["\mu_1^\partial"] & \underbrace{\big(\Omega^1(\partial M,\LV)\oplus\Omega^2(\partial M,\LV)\big)}_{\in\,(\alpha,\beta)}\rar["\mu_1^\partial"] & \underbrace{\Omega^3(\partial M,\LV)}_{\in\,\alpha^+}.
            \end{tikzcd}
        \end{equation}
        where the first row represents $\LV$ and the second row $\LV_\partial$. Notice that there is the extra component $\Omega^2(\partial M,\mathfrak{g})$ in the boundary $L_\infty$-algebra $\LV_\partial$ labelled by $\beta$). The reason for this is that only the gauge potential $A\in\Omega^1(M,\LV)$ appears with second-order terms in the Yang--Mills action. Furthermore, upon letting `$*$' be the Hodge operator with respect to the metric $g$, we have (see e.g.~\cite{Zeitlin:2007yf,Zeitlin:2008cc,Jurco:2018sby})
        \begin{equation}
            \begin{gathered}
                \mu_1(c)\coloneqq\mathrm d_Mc,
                \quad
                \mu_1(A)\coloneqq\mathrm d_M{*\mathrm d_MA},
                \quad
                \mu_1(A^+)\coloneqq\mathrm d_MA^+,
                \\
                \mu_2(c,c')\coloneqq[c,c']_\LV,
                \quad
                \mu_2(c,A)\coloneqq[c,A]_\LV,
                \quad
                \mu_2(c,A^+)\coloneqq[c,A^+]_\LV,
                \quad
                \mu_2(c,c'^+)\coloneqq[c,c'^+]_\LV,
                \\
                \mu_2(A,A')\coloneqq\mathrm d_M{*[A,A']_\LV}+[A,{*\mathrm d_MA'}]_\LV+[A',{*\mathrm d_MA}]_\LV,
                \quad
                \mu_2(A,A'^+)\coloneqq[A,A'^+]_\LV,
                \\
                \mu_3(A,A',A'')\coloneqq[A,{*[A',A'']_\LV}]_\LV+\text{cyclic}
            \end{gathered}
        \end{equation}
        as well as 
        \begin{equation}
            \begin{gathered}
                \mu_1^\partial(\gamma)\coloneqq
                \begin{pmatrix}
                    \mathrm d_{\partial M}\gamma\\ 0
                \end{pmatrix},
                \quad
                \mu_1^\partial
                \begin{pmatrix}
                    \alpha\\ \beta
                \end{pmatrix}
                \coloneqq\mathrm d_{\partial M}\beta,
                \\
                \mu_2^\partial(\gamma,\gamma')\coloneqq[\gamma,\gamma']_\LV,
                \quad
                \mu_2^\partial
                \left(
                    \gamma,
                    \begin{pmatrix}
                        \alpha\\ \beta
                    \end{pmatrix}
                \right)
                \coloneqq
                \begin{pmatrix}
                    [\gamma,\alpha]_\LV\\ [\gamma,\beta]_\LV
                \end{pmatrix},
                \quad
                \mu_2^\partial(\gamma,\alpha^+)\coloneqq[\gamma,\alpha^+]_\LV,
                \\
                \mu_2^\partial
                \left(
                    \begin{pmatrix}
                        \alpha\\ \beta
                    \end{pmatrix},
                    \begin{pmatrix}
                        \alpha'\\ \beta'
                    \end{pmatrix}
                \right)
                \coloneqq[\alpha,\beta',]_\LV+[\alpha',\beta]_\LV
            \end{gathered}
        \end{equation}
        and
        \begin{equation}
            \begin{gathered}
                \pi_1(c)\coloneqq c|_{\partial M}, 
                \quad
                \pi_1(A)\coloneqq
                \left.
                    \begin{pmatrix}
                        A\\ *\mathrm d_MA+\theta\mathrm d_MA
                    \end{pmatrix}
                \right|_{\partial M},
                \quad
                \pi_1(A^+)\coloneqq A^+|_{\partial M},
                \\
                \pi_2(A,A')\coloneqq
                \left.
                    \begin{pmatrix}
                        0\\ {*[A,A']_\LV}+\theta[A,A']_\LV
                    \end{pmatrix}
                \right|_{\partial M}.
            \end{gathered}
        \end{equation}
        In addition, we introduce the bilinear forms that have the non-vanishing components
        \begin{equation}
            \begin{gathered}
                \langle c,c'^+\rangle\coloneqq\int_M\langle c,c'^+\rangle_\LV\eqqcolon\langle c'^+,c\rangle_\LV,
                \quad
                \langle A,A'^+\rangle\coloneqq\int_M\langle A,A'^+\rangle_\LV\eqqcolon\langle A'^+,A\rangle_\LV,
                \\
                \langle\gamma,\alpha^+\rangle_{\partial}\coloneqq\int_{\partial M}\langle\gamma,\alpha^+\rangle_\LV,
                \quad
                \left\langle
                    \begin{pmatrix}
                        \alpha\\ \beta
                    \end{pmatrix},
                    \begin{pmatrix}
                        \alpha'\\ \beta'
                    \end{pmatrix}
                \right\rangle_{\partial}
                \coloneqq\int_{\partial M}\langle\alpha,\beta'\rangle_\LV.
            \end{gathered}
        \end{equation}
    \end{subequations}
    Then, the relative homotopy Maurer--Cartan action~\eqref{eq:relative_MC_action} becomes
    \begin{equation}
        \begin{aligned}
            S&=\frac12\int_M\Big\{\langle A,\mathrm d_M{*\mathrm d_MA}\rangle_\LV+\langle A^+,\mathrm d_Mc\rangle_\LV-\langle c,\mathrm d_MA^+\rangle_\LV\Big\}
            \\
            &\kern1cm+\int_M\Big\{\tfrac1{3!}\langle A,\mathrm d_M{*[A,A]_\LV}+2[A,{*\mathrm d_MA}]_\LV\rangle_\LV+\langle A^+,[A,c]_\LV\rangle_\LV+\tfrac12\langle c^+,[c,c]_\LV\rangle_\LV\Big\}
            \\
            &\kern1.5cm+\frac18\int_M\langle A,[A,{*[A,A]_\LV}]_\LV\rangle_\LV
            \\
            &\kern2cm+\frac12\int_{\partial M}\Big\{\langle c,A^+\rangle_\LV+\langle A,{*\mathrm d_MA}+\theta\mathrm d_MA\rangle_\LV\Big\}\Big|_{\partial M}
            \\
            &\kern2.5cm+\frac1{3!}\int_{\partial M}\langle A,{*[A,A]}_\LV+\theta[A,A]_\LV\rangle_\LV\Big|_{\partial M}
            \\
            &=\int_M\Big\{\tfrac12\langle F,{*F}\rangle_\LV+\langle\nabla_MA^+,c\rangle_\LV+\tfrac12\langle c^+,[c,c]_\LV\rangle_\LV+\tfrac{\theta}{2}\langle F,F\rangle_\LV\Big\},
        \end{aligned}
    \end{equation}
    where, as before, $F\coloneqq\mathrm d_MA+\frac12[A,A]_\LV$ and $\nabla_Mc\coloneqq\mathrm d_Mc+[A,c]_\LV$.

    \addcontentsline{toc}{section}{Acknowledgements}
    \section*{Acknowledgements}

    We thank Branislav Jur{\v c}o and Christian Saemann for fruitful discussions.

    \addcontentsline{toc}{section}{Declarations}
    \section*{Declarations}

    \textbf{Funding.}
    L.~A., H.~K., and C.~A.~S.~Y.~gratefully acknowledge the financial support of the Leverhulme Trust, Research Project Grant number \textsc{rpg}-2021-092.\\[5pt]
    \textbf{Conflict of interest.}
    The authors have no relevant financial or non-financial interests to disclose.\\[5pt]
    \textbf{Data statement.}
    No additional research data beyond the data presented and cited in this work are needed to validate the research findings in this work.\\[5pt]
    \textbf{Licence statement.}
    For the purpose of open access, the authors have applied a Creative Commons Attribution (CC-BY) licence to any author-accepted manuscript version arising.
    
    \addcontentsline{toc}{section}{References}
    \bibliographystyle{unsrturl}
    \bibliography{bigone,extra}

\end{document}